\documentclass[11pt]{article}

\usepackage{amsfonts,amsmath,epsfig, graphicx, fullpage,amsthm}
\usepackage{color,bm}
\usepackage{multirow}

\usepackage{comment}
\excludecomment{full}
\includecomment{extended}

\excludecomment{fullversion}

\newcommand{\junote}[1]{}

\newcommand{\vertical}[1]{{\left\vert\kern-0.25ex\left\vert\kern-0.25ex\left\vert #1  \right\vert\kern-0.25ex\right\vert\kern-0.25ex\right\vert}}

\newcommand{\bA}{\mathbf{A}}
\newcommand{\bB}{\mathbf{B}}
\newcommand{\bC}{\mathbf{C}}
\newcommand{\bE}{\mathbf{E}}
\newcommand{\bH}{\mathbf{H}}
\newcommand{\bN}{\mathbf{N}}
\newcommand{\bY}{\mathbf{Y}}

\renewcommand{\bH}{\mathbf{H}}
\newcommand{\bW}{\mathbf{W}}
\newcommand{\bD}{\mathbf{D}}
\newcommand{\bT}{\mathbf{T}}
\newcommand{\bZ}{\mathbf{Z}}

\newcommand{\p}{\mathsf{Pr}}

\newcommand{\cD}{\mathcal{D}}

\newcommand{\cN}{\mathcal{N}}

\newtheorem{myremark}{Remark} 
\newenvironment{remark}{\begin{myremark}}{\end{myremark}}

\newtheorem{myexample}{Example}

\newcommand{\set}[1]{\left\{ {#1} \right\}}
\newcommand{\paren}[1]{\left( {#1} \right)}
\newcommand{\sparen}[1]{\left[ {#1} \right]}

\DeclareMathOperator{\poly}{poly}

\newcommand{{\priv}}{\mathsf{Priv}}

\newcommand{{\bQ}}{\mathbf{Q}}
\newcommand{{\bR}}{\mathbf{R}}
\newcommand{{\bU}}{\mathbf{U}}
\newcommand{{\bV}}{\mathbf{V}}
\newcommand{\R}{\mathbb{R}}

\newtheorem{theorem}{Theorem}
\newtheorem{lemma}[theorem]{Lemma}

\newtheorem{claim}[theorem]{Claim}

\newtheorem{fact}[theorem]{Fact}
\newtheorem{corollary}[theorem]{Corollary}

\newtheorem{prob}{Problem}

\newtheorem{definition}[theorem]{Definition} 

\newcommand{\thmref}[1]{Theorem~\ref{thm:#1}}

\newcommand{\lemref}[1]{Lemma~\ref{lem:#1}}

\newcommand{\corref}[1]{Corollary~\ref{cor:#1}}
\newcommand{\claimref}[1]{Claim~\ref{claim:#1}}

\newcommand{\secref}[1]{Section~\ref{sec:#1}}

\newcommand{\appref}[1]{Appendix~\ref{app:#1}}

\newcommand{\figref}[1]{Figure~\ref{fig:#1}}

\newcommand{\eqnref}[1]{equation~(\ref{eq:#1})}

\newcommand{\bLambda}{\mathbf{\Lambda}}
\newcommand{\bSigma}{\mathbf{\Sigma}}
\newcommand{\bPsi}{\mathbf{\Psi}}
\newcommand{\bPhi}{\mathbf{\Phi}}
\newcommand{\bPi}{\mathbf{\Pi}}
\newcommand{\bOmega}{\mathbf{\Omega}}

\newcommand{\cov}{\mathsf{COV}}
\newcommand{\lra}{\mathsf{LRA}}
\newcommand{\PDF}{\mathsf{PDF}}

\newcommand{\rad}{\mathsf{Rad}}
\newcommand{\lrf}{\mathsf{LRF}}

\newcommand{\argmin}{\operatornamewithlimits{argmin}}
\newcommand{\bS}{\mathbf{S}}

\newcommand{\bX}{\textbf{X}}
\newcommand{\by}{\textbf{y}}
\newcommand{\bx}{\textbf{x}}

\newcommand{\E}{\mathbb{E}}

\newcommand{\I}{\mathbb{I}}

\newcommand{\N}{\mathbb{N}}

\usepackage{tabularx} 

\usepackage[bottom]{footmisc}

\newcommand{\mypar}[1]{\vspace{10pt} \noindent {\bf {#1}:}}
\newcommand{\contrib}{\mypar}

\newcommand{\naiveadditive}{ \paren{ \frac{ \sqrt{m} +\sqrt{n} + \alpha^{-2} }{(1-\alpha)^{3/2}\varepsilon} } } 
\newcommand{\naivecor}{ {\paren{ \sqrt{m} +\sqrt{n}}}{\varepsilon}^{-1}}  
\newcommand{\naivecontinual}{ \naivecor \log T}

\newcommand{\spaceadditive}{  \paren{\frac{k\alpha^{-1 } \ln(1/\delta) +\sqrt{m } + \sqrt{n}}{ \varepsilon (1-\alpha)^{3}}} } 
\newcommand{\spacecor}{ {\paren{\sqrt{m } + \sqrt{n} }}{\varepsilon}^{-1} } 
\newcommand{\spacecontinual}{ \spacecor \log T }

\newcommand{\naivetheorem}{ 	Let  $m, n \in \N$ and $\varepsilon,\delta$ be the input parameters with $m \geq n$.  Let $k$ be the desired rank of the factorization.  Let $\alpha \in (0,1)$ be an arbitrary constant. Given an $m \times n$ matrix $\bA$ in the turnstile update model,  the algorithm {\scshape Spectral-LRF}, presented in~\figref{private}, satisfies the following properties:
\begin{enumerate}
	 \item Let $\eta=\max\set{k^2,1/\alpha}$. Then {\scshape Spectral-LRF} uses $O((m \eta \alpha^{-1} + n  \eta \alpha^{-3} ) \log (k/\alpha^2) \log(1/\delta))$ space.
	\item {\scshape Spectral-LRF} is $(\varepsilon,\delta)$-differentially private under $\priv_1$. %
	\item With probability at least $9/10$ over the coin tosses of {\scshape Spectral-LRF}, 
	 \begin{align*} &\| \bA - \mathbf{M}_k \|_2  \leq \frac{(1+\alpha)}{(1-\alpha)^2}\Delta_k(\bA) 
	 	\quad + {O}\paren{ \naiveadditive \sqrt{\log (1/\delta)}}, \end{align*} 
		$~\text{where}~\mathbf{M}_k = \bU_k \bSigma_k \bV_k^{\mathsf T}.$
\end{enumerate}
}

\newcommand{\spacetheorem}{
	Let  $m, n,k \in \N$ and $\varepsilon,\delta$ be the input parameters.  Let $\alpha \in (0,1)$ be an arbitrary constant, $\kappa=(1+\alpha)/(1-\alpha)$, $\sigma_{\mathsf min}=16 \ln(1/\delta) \sqrt{t \kappa \ln(4/\delta)}/\varepsilon$,  $\rho_1={\sqrt{(1+\alpha) \ln(1/\delta)}}/{\varepsilon}$, and $\rho_2=\rho_1\sqrt{1+\alpha}$. Given an $m \times n$ matrix $\bA$ in the turnstile update model,  the algorithm {\scshape Low-Space-LRF}, presented in~\figref{space}, satisfies the following properties:
\begin{enumerate}
	 \item Let $\eta=\max\set{k^2,1/\alpha}$. Then {\scshape Spectral-LRF} uses $O((m + n ) \eta \alpha^{-1} \log (k/\alpha) \log(1/\delta))$ space. 
	\item {\scshape Spectral-LRF} is $(\varepsilon,\delta)$-differentially private  under $\priv_2$. \label{spaceprivacy}%
	\item With probability at least $9/10$ over the coin tosses of {\scshape Spectral-LRF}, \label{spacecorrectness}
\begin{align*} & \| \mathbf{M}_k - \begin{pmatrix} \bA & \mathbf{0} \end{pmatrix} \|_2  \leq  \frac{(1+\alpha)^2}{(1-\alpha)^4}  \Delta_k(\bA) 
 + O \paren{ \spaceadditive{\sqrt{\ln(1/\delta)}} }.   \end{align*}
\end{enumerate}
}

\renewcommand{\bH}{\mathbf{H}}

\newcommand{\sT}{\mathsf{T}}

\renewcommand{\t}[1]{\widetilde{#1}}
\newcommand{\h}[1]{\widehat{#1}}

\usepackage[toc]{multitoc}

\begin{document}
\title{On Low-Space Differentially Private Low-rank Factorization in the Spectral Norm}
\author{
Jalaj Upadhyay \\ College of Information Science and Technology \\ Pennsylvania State University \\ \small{\sf jalaj@psu.edu}
}
\date{}
\maketitle

\pagenumbering{gobble}
\begin{abstract}
Low-rank factorization is used in many areas of computer science where one performs spectral analysis on large sensitive data  stored in the form of matrices.  In this paper, we study differentially private low-rank factorization of a matrix with respect to the spectral norm in the  turnstile update model. In this problem, given an input matrix $\bA\in \R^{m \times n}$ updated in the turnstile manner and a target rank $k$, the goal is to find  two rank-$k$ orthogonal matrices $\bU_k \in \R^{m \times k}$ and $ \bV_k \in \R^{n \times k}$, and one positive semidefinite diagonal matrix $\bSigma_k \in \R^{k \times k}$ such that $\bA \approx \bU_k \bSigma_k \bV_k^\sT$ with respect to the spectral norm. 

Our main contributions are two computationally efficient and sub-linear space algorithms for computing a differentially private low-rank factorization. We consider two levels of privacy. In the first level of privacy, we consider two matrices neighboring if their difference has a Frobenius norm at most $1$. In the second level of privacy, we consider two matrices as neighboring if their difference can be represented as an outer product of two unit vectors. Both these privacy levels are stronger than those studied in the earlier papers such as Dwork {\it et al.} (STOC 2014), Hardt and Roth (STOC 2013), and Hardt and Price (NIPS 2014). 

As a corollary to our results, we get  non-private algorithms that compute low-rank factorization in the turnstile update model with respect to the spectral norm. We note that, prior to this work,   no algorithm that outputs low-rank factorization with respect to the spectral norm in the turnstile update model was known; i.e., our algorithm gives the first non-private low-rank factorization with respect to the spectral norm in the turnstile update mode.

Our algorithms generate  private linear sketches of the input matrix. Therefore, using the binary tree mechanism of Chan {\it et al.} (TISSEC: 14(3)) and Dwork {\it et al.} (STOC 2010), we  get  algorithms for continual release of low-rank factorization under both these privacy levels. This gives the first instance of  differentially private algorithms with continual release that guarantees a stronger level of privacy than  event-level privacy.

\end{abstract}



\pagenumbering{arabic}
\section{Introduction} \label{sec:introduction}
Spectral analysis of matrices  is  used in many areas where large {\em sensitive data} are stored in the form of matrices. A partial list includes  data mining~\cite{AFKMS01}, recommendation systems~\cite{DKR02}, 
information retrieval~\cite{PRTV00,SC04}, 
web search~\cite{AFKM01,Kleinberg99}, 
clustering~\cite{CEMMP15,DFKVV04,McSherry01}, 
and learning  distributions~\cite{AM05,KSV05}.   A general technique  used in the literature involving such analysis  first computes an approximate low-rank factorization ($\lrf$) of the matrix to avoid the curse of dimensionality with the hope that the output from using the $\lrf$ does not ``differ by much" from the output from using the original data matrix. These algorithms often do not consider privacy (or use ad-hoc anonymization methods) and do not account for dynamic data matrices. 
On the other hand, the limitations of ad-hoc anonymization of these data  has been exemplified by the denonymization of the Netflix datasets~\cite{NS08} and the real-world data changes dynamically. This raises two natural questions:
 
\begin{quote}
{\bf (i)} Can we {\em efficiently} perform $\lrf$ under the spectral norm while {\em preserving the privacy} of the dynamic  data matrix? 

{\bf (ii)} Can we  release $\lrf$ under the continual release model~\cite{DNPR10}?
\end{quote}

We answer these questions in the  affirmative. Specifically, we give  two {\em differentially private} algorithms that receive an $m \times n$ matrix  in the {\em  turnstile update model} and output a $k$-rank factorization with respect to the spectral norm. Our algorithms efficiently compute a small space {\em linear sketches} of the private matrix. We use the private matrix only when generating the linear sketches; therefore, our algorithms can be easily modified  to get differentially private algorithms under {\em continual release} (i.e., the output is released at every time epoch) by using the binary tree mechanism~\cite{CSS,DNPR10}.

In the recent past, there has been a lot of attention to a related problem of differentially private {\em low-rank approximation} ($\lra$) with respect to the spectral norm~\cite{DTTZ14,HP14,HR13,KT13} leading to an optimal approximation error. However, these algorithms run in $O(mnk)$ time, use $O(mn)$ space, and assume that the matrix is static. 
It is easy to see that, if there are no time or space constraints, then $\lra$ can be used to give $\lrf$ in polynomial time. However, it is not known how to efficiently factorize a matrix in sub-linear space. Even in the non-private setting, there is no known algorithm that computes spectral approximation in the turnstile update model --- the only algorithm that computes an $\lra$ is under Frobenius norm by Boutisidis {\it et al.}~\cite{BWZ16} and Clarkson and Woodruff~\cite{CW13}. Unfortunately, the algorithm of Boutsidis {\it et al.}~\cite{BWZ16} is not robust against noise.  In a closely related work, Upadhyay~\cite{Upadhyay16} gave a  differentially private algorithm for $\lrf$ in the turnstile update model when the approximation metric is the Frobenius norm. This algorithm can be seen as a robust version of the algorithm of Boutsidis {\it et al.}~\cite{BWZ16}. 

One natural point to start would be to see if one could use the existing algorithm that gives approximation under Frobenius norm.
Unfortunately, Karnin and Liberty~\cite{KL14} observed that even though the exact solution of $\lrf$ is the same with respect to both the Frobenius and the spectral norm, the same cannot be said about their approximate solutions. Moreover, even though the spectral norm and Frobenius norm are trivially within the factor of each other, the same cannot be said about their approximate solutions. 
\junote{Compare with BWZ}
Therefore, it is not clear if we can use the algorithms of Upadhyay~\cite{Upadhyay16} and Boutsidis {\it et al.}~\cite{BWZ16} to give algorithms for $\lrf$ with respect to the spectral norm.   \junote{Make a note to appendix where this is explained. Talk about personal communication with Woodruff.}
On the other hand, there is an instance of algorithm by Halko {\it et al.}~\cite{HMT11} that gives approximation under both the spectral as well as the Frobenius norm. However, their algorithm is not in a turnstile update model. 

\subsection{Formal Problem Description} \label{sec:problems}
In this paper, we study { differentially private $\lrf$} in the {  turnstile update model}  when the approximation metric is the spectral norm. In the  turnstile update model, an update is in the form of triples $\{ i,j,s_\tau\}$, where $1 \leq i \leq m, 1 \leq j \leq n,s_\tau \in \R$ for all $\tau\geq 1.$  This leads to a change in the $(i,j)$-th entry of the private matrix $\bA$ as follows: $\bA_{i,j} \leftarrow \bA_{i,j} + s_\tau.$ 
We first  give the formal definition of differential privacy. 
\begin{definition} \label{defn:approxdp}
	A randomized algorithm $\mathfrak{M}$ gives {\em $(\varepsilon, \delta)$-differential privacy}, if for all neighbouring databases (presented in the form of matrices) $\bA$ and ${\bA}'$, and all subsets $S$ in the range of $\mathfrak{M}$, 	$ \p[\mathfrak{M}(\bA) \in S] \leq  e^\varepsilon \p[\mathfrak{M}({\bA}') \in S] + \delta $ over the coin tosses of $\mathfrak{M}$. 
\end{definition}

We can now formally define the main problem we address in this paper.
\begin{prob} 
\label{prob:private_factor} ($(\varepsilon,\delta, \beta,\gamma,\zeta,k)\mbox{-}\lrf$).
Given parameters $\varepsilon, \delta, \gamma,\beta, \zeta$, the target rank $k$, and a private $m \times n$ matrix $\bA$  updated in the general turnstile update model, output  (under continual release) an $(\varepsilon,\delta)$-differentially private rank-$k$ factorization $\widetilde{\bU}_k, \widetilde{\bSigma}_k, \widetilde{\bV}_k$ such that 
\begin{align*}  \p \sparen{ \| \bA -\widetilde{\bU}_k \widetilde{\bSigma}_k \widetilde{\bV}_k^\sT \|_2 \leq \gamma \| \bA- [\bA]_k \|_2 +\zeta } \geq 1-\beta,\end{align*} 
where $\| \cdot \|_2$ denotes the spectral norm and $[\bA]_k$ is the best rank-$k$ approximation of $\bA$. \junote{Compare with PCA.}
We refer the term $\gamma$ as the {\em multiplicative error} and the term $\zeta$ as the {\em additive error}. We call the tuples $(\varepsilon,\delta)$ the {\em privacy parameters}. 
\end{prob}

Our problem statement is more general than principal component analysis in the sense that we require approximation to both the left and the singular vectors. We discuss this in more details in~\appref{related}.

\contrib{Granularity of Privacy}
In the past, differentially private $\lra$  with respect to the spectral norm have been proposed with varying levels of privacy. Kapralov and Talwar~\cite{KT13} considered two matrices neighboring if the difference of their spectral norm is at most $1$. Dwork {\it et al.}~\cite{DTTZ14} considered two row-normalized matrices neighboring if they differ in one row. Hardt and Price~\cite{HP14} and Hardt and Roth~\cite{HR13} considered two matrices as neighboring if they differ exactly in one entry. In this paper, we consider even stronger levels of privacy. This matrices deals with streaming matrices; therefore, for the sake of discussion and ease of comparison of privacy model, below we define the granularity of privacy with respect to neighboring matrices. Two streams are said to be neighboring if they are formed using neighboring matrices. 

\junote{Give precise privacy definition in the terms of streams.}
In the first privacy level, $\priv_1$, we call two matrices $\bA$ and $\bA'$ as neighboring if  $\|\bA - \bA' \|_F \leq 1$, where $\| \cdot \|_F$ denotes the Frobenius norm. 
In the second privacy level, $\priv_2$, we consider two matrices are neighboring  if  $\|\bA - \bA' \|_F \leq 1$ and $\bA - \bA'$ is a rank-$1$ matrix. In other words, the matrices differ in only one spectrum by at most $1$.  Both these privacy notions are motivated by natural scenarios that are not captured by the privacy level studied in previous works~\cite{DTTZ14,HP14,HR13}. For example, $\priv_2$ is a natural choice where the spectrum  of the input matrix is the key feature, which is the case in scenarios where spectral analyses are performed. The choice of $\priv_1$ is motivated by the examples listed in Upadhyay~\cite{Upadhyay16}, where the presence or absence of an individual in a social graph can lead to a change of at most $1$ in the Frobenius norm of the corresponding adjacency matrix. We refer the readers to Upadhyay~\cite{Upadhyay16} for more details. Another notable example where $\priv_2$ make sense is the {\sf word2vec} model. \junote{Use word2vec example listed by a reviewer.}

\subsection{Contributions of This Paper} \label{sec:contributions}
All the earlier works~\cite{HP14,HR13,KT13}  use an approach known as the {\em subspace iteration}. A subspace iteration algorithm runs for $k$ rounds. In every iteration, the algorithm computes (and stores) the top singular vector of the matrix, say $\mathbf{v}$, and updates the matrix $\bA$ as follows: $\bA \leftarrow \bA - (\mathbf{v}^\sT \bA \mathbf{v})\mathbf{v}\mathbf{v}^\sT$. Our key insight is that we can compute the top-$k$ singular vector in one step (even in the turnstile update model) by using linear sketches. This also allows us to reduce the additive error as we do not have to add noise $k$ times to preserve the privacy in every iteration, giving us $\sqrt{k}$ improvement over  private algorithms that use subspace iteration. 

Throughout this section, we assume $\delta = o(1/n^2)$ and let $\widetilde O (\cdot)$ hides a ${\log n}$ factor  and denote by $\Delta_k(\bA):=\| \bA - [\bA]_k \|_2.$ 
We state our results assuming $m \geq n$. All our results hold true for $m <n$ with the roles of $m$ and $n$ reversed. Our algorithms are also efficient because all the costly computations are done on low-dimensional sketches and  the sketches themselves can be generated and updated efficiently using known techniques~\cite{CW13,KN14}. 

\contrib{$\lrf$ With Respect to Spectral Norm Under $\priv_1$} \label{sec:informalnaive}
Our first algorithm, {\scshape Spectral-LRF}, is based on the following intuition.  Let $[\bU]_k [\bSigma]_k [\bV]_k^{\mathsf T}$ be a singular-value decomposition of $[\bA]_k$. Therefore, if an orthonormal column matrix $\bU$ is a ``faithful" representation of  $[\bU]_k$, then $\t{\bX}:=[\bSigma]_k [\bV]_k^{\mathsf T} \approx \argmin_{r(\bX)\leq k} \| \bU \bX - \bA \|$  
and $\| \bU \t{\bX}  - \bA\|_2 \approx \| [\bA]_k - \bA\|_2$. We show that the matrix $\bU$ can be computed from a linear sketch. If we pick $\bS$ to be a random projection matrix such that, simultaneously for all $\bX$, $\| \bS(\bU {\bX}   - \bA )\|_2 \approx \| (\bU \bX   - \bA )\|_2 $ with high probability, then $\| \bS(\bU \t{\bX}   - \bA )\|_2 \approx  \| [\bA]_k - \bA\|_2$ with high probability. 
In other words, we need to store $\bS \bA$ and the linear sketch used to compute $\bU$, and output the product of  $\bU$ and $\argmin_{r(\bX)\leq k} \| \bS(\bU \bX  - \bA)\|_2$. 
We show the following.

\begin{theorem} (Informal statement of~\thmref{naive}). \label{thm:informalnaive}
 	Let  $m, n \in \N$ (where $m \geq n$), $(\varepsilon,\delta)$ be the privacy parameters, and $k$ be the desired rank of the factorization.  Let $\alpha \in (0,1)$ be an arbitrary  constant and $\eta=\max \set{k^2,\alpha^{-1}}$. Given an $m \times n$ matrix $\bA$ in the turnstile update model,  there exists an efficient	 $(\varepsilon,\delta)$-differentially private algorithm that uses $O_\delta(m\eta \alpha^{-1} + n  \eta  \alpha^{-3})$ space and outputs a $k$-rank factorization ${\bU}_k, {\bSigma}_k, {\bV}_k$ such that,  with probability at least $9/10$ over the coin tosses of the algorithm,
\begin{align*}   \| \bA -{\bU}_k {\bSigma}_k {\bV}_k^\sT \|_2 &\leq \frac{(1 +\alpha )}{(1-\alpha)^2} \Delta_k(\bA) 
+ \widetilde O  \paren{ \naiveadditive   }  .\end{align*} 
\end{theorem}

In practice, $k \ll \max \set{m,n}$ and $\alpha$ is  a small constant. In that case, we  have $(1-\alpha)^{-1} \approx (1+\alpha)$. If we scale the value of $\alpha$ appropriately, then we  have the following corollary to~\thmref{informalnaive}.

\begin{corollary}  (Informal statement of~\corref{spectral}). \label{cor:informalspectral}
Let $m,n,\varepsilon,\delta$ be as in~\thmref{informalnaive}  with $m \geq n$. Then there exists an efficient 	 $(\varepsilon,\delta)$-differentially private algorithm under $\priv_1$ that  solves $(\varepsilon,\delta, \gamma,1/10,\zeta,k)\mbox{-}\lrf$ in $O_\delta(m\eta \alpha^{-1} + n \eta  \alpha^{-3})$ space, where  $\gamma:= {(1 +\alpha )}$ {and} $ \zeta:= \widetilde O \paren{ \naivecor }.$
\end{corollary}

\begin{remark} \junote{Optimality wrt additive error. Can we claim this?}
 Hardt and Roth~\cite[Thm 1.2]{HR13} showed that any differentially private algorithm incurs an additive error  $\zeta=\Omega(\varepsilon^{-1}\sqrt{n})$ for an $n \times n$  matrix. Their lower bound holds even when the  algorithm can access the input matrix any number of times. \corref{informalspectral} shows that we can achieve the same bound, up to a logarithmic factor, even when the   matrix is updated in the turnstile manner at the cost of small multiplicative error that depends only on the $(k+1)$-th singular value.
\end{remark}

\contrib{Improving the Space Bound Under $\priv_2$}  The  algorithm {\scshape Spectral-LRF} had a space requirement that depends on the dimension of both $\bS$ and $\bU$. 
However,  if the matrix is almost square, then the storage required to store $\bS \bA$ is much higher than that to store the sketch that is used to compute $\bU$. 
A direct observation then is that, if we can simultaneously generate an orthonormal matrix $\bV$ which is a ``faithful" representation of $[\bV]_k$, then $\t{\bX}:=[\bSigma]_k \approx \argmin_{\mathsf{r}(\bX) \leq k} \| \bU \bX \bV^\sT - \bA \|_2.$ Further, if we pick $\mathbf{Q}$ and $\mathbf{R}$ that satisfy similar properties as $\bS$ in~\secref{informalnaive}, then $\| \mathbf{Q}(\bU \t{\bX} \bV^\sT  - \bA ) \mathbf{R} \|_2 \approx  \| [\bA]_k - \bA\|_2$
Therefore, we can output the product of matrices $\bU$, $\argmin_{r(\bX)\leq k} \| \mathbf{Q}(\bU \t{\bX} \bV^\sT  - \bA) \mathbf{R} \|_2$, and $\bV^\sT$. This forms the basis of {\scshape Low-Space-LRF}. 

\begin{theorem} (Informal statement of~\thmref{space}). \label{thm:informalspace}
 	Let  $m, n \in \N$ (where $m \leq n$) and $\varepsilon,\delta$ be the input parameters.  Let $k$ be the desired rank of the factorization.  Let $0<\alpha <1$ be an arbitrary  constant and $\eta=\max \set{k^2,\alpha^{-1}}$. Given an $m \times n$ matrix $\bA$   in the turnstile update model,  there exists an efficient	 $(\varepsilon,\delta)$ differentially private algorithm that  outputs a $k$-rank factorization ${\bU}_k, {\bSigma}_k ,{\bV}_k$ using $O_\delta((m+n)\eta \alpha^{-1} \log(1/\delta) )$ space, such that with probability at least $9/10$ over the coin tosses of the algorithm,
\begin{align*}   \| \bA -{\bU}_k {\bSigma}_k {\bV}_k^\sT \|_2 &\leq \frac{(1 +\alpha )^2}{(1-\alpha)^4} \Delta_k(\bA) 
 + \widetilde O  \spaceadditive    .\end{align*}  
\end{theorem}

As before we get analogous results to~\corref{informalspectral}. 
\begin{corollary}  (\corref{space}, informal). \label{cor:informalspace}
Let $m,n,\varepsilon,\delta$ be as in~\thmref{informalnaive} and $\alpha \in (0,1)$ be a small constant. Then there exists an 	efficient $(\varepsilon,\delta)$ differentially private algorithm that   solves $(\varepsilon,\delta, \gamma,1/10,\zeta,k)\mbox{-}\lrf$ using  $O((m+n)\eta \alpha^{-1} \log(1/\delta) )$ space, where  $\gamma:= {(1 +\alpha )}$ {and} $ \zeta:= \widetilde O \paren{\spacecor}.$
\end{corollary}
\junote{Talk about space complexity and why it is so high?}
\begin{remark}
\label{remark:space}
The dependency of space complexity on $\alpha$ is unavoidable. To see this, consider a matrix $\bA$ of rank $k+1$. In this case, approximation with respect to the Frobenius norm and spectral norm are equivalent. However, due to the result of Upadhyay~\cite{Upadhyay16}, any algorithm that achieves $(1+\alpha)$ multiplicative error and $O(n)$ additive error has to use $\Omega(nk/\alpha)$ space. Therefore, one cannot hope to get an $\lra$ under spectral norm with no multiplicative error and low space under the turnstile update model. 
\end{remark}

\subsection{Applications of Our Results}
\contrib{Application in Continual Release} One of the characteristics of our algorithms is that  $\bU$ is computed from a linear sketch of the private input matrix and $\bS \bA$ is a linear sketch. This implies that we can use the binary tree mechanism ~\cite{CSS,DNPR10} to get a low-rank factorization under continual release by paying an $O(\log T)$ factor in the additive error, where $T$ is the number of updates. This is stated in the form of the following results (for a small constant $\alpha \in (0,1)$). 
\begin{theorem}  (Informal statement of~\thmref{continual}). 
\label{thm:informalcontinual}
Let $m,n,\varepsilon,\delta$ be as in~\thmref{informalnaive}. Then there exists an efficient	 $(\varepsilon,\delta)$-differentially private algorithm under $\priv_1$ that   solves $(\varepsilon,\delta, \gamma,1/10,\zeta,k)\mbox{-}\lrf$  under continual release for $T$ time epochs, where   $\gamma:= {(1 +\alpha )}$ {and} $ \zeta:= \widetilde O  \paren{ \naivecontinual }.$
\end{theorem}
\begin{theorem}  (Informal statement of~\thmref{spacecontinual}).  \label{cor:informalspacecontinual}
Let $m,n,\varepsilon,\delta$ be as in~\thmref{informalspace}. Then there exists an efficient	 $(\varepsilon,\delta)$-differentially private algorithm under $\priv_2$ that   solves $(\varepsilon,\delta, \gamma,1/10,\zeta,k)\mbox{-}\lrf$  under continual release for $T$ time epochs, where   $\gamma:= {(1 +\alpha )}$ {and} $ \zeta:= \widetilde O \paren{ \spacecontinual }.$
\end{theorem}


\contrib{Application in Non-private Setting}
As an immediate result of~\thmref{informalnaive}, when $\varepsilon \rightarrow \infty$, we get a non-private algorithm that achieves low-rank factorization in the turnstile update model with respect to the spectral norm. Prior to this work, no algorithm that outputs $\lra$ with respect to the spectral norm in the turnstile update model was known.

\begin{theorem} 
\label{thm:informalnonprivate}
Let $m,n$ be as in~\thmref{informalnaive} and $\alpha \in (0,1)$. Given a matrix $\bA \in \R^{m \times n}$  in the turnstile update model, there exists an efficient algorithm that uses $O(mk \alpha^{-1} + n   \alpha^{-4})$ space and outputs a $k$-rank factorization ${\bU}_k, {\bSigma}_k, {\bV}_k$ such that  $\| \bA - {\bU}_k {\bSigma}_k {\bV}_k^\sT \|_2 \leq  {(1 +\alpha )} \Delta_k(\bA)$. 
\end{theorem}

Similarly, we can improve the space bound of the algorithm by setting $\varepsilon \rightarrow \infty$ in~\thmref{informalnaive}. 
\begin{theorem} 
 \label{thm:informalspacenonprivate}
Let $m,n$ be as in~\thmref{informalnaive} and $\alpha \in (0,1)$. Given a matrix $\bA \in \R^{m \times n}$ in the turnstile update model, there exists an efficient algorithm that uses $O((m + n)k \alpha^{-1} )$ space and outputs a $k$-rank factorization ${\bU}_k, {\bSigma}_k, {\bV}_k$, such that  $\| \bA - {\bU}_k{\bSigma}_k{\bV}_k^\sT \|_2 \leq  {(1 +\alpha )} \Delta_k(\bA)$. 
\end{theorem}

\subsection{Technical Overview} \label{sec:overview}
In this section, we only present the proof idea for the correctness of  {\scshape Spectral-LRF}. The proof idea for  the correctness of  {\scshape Low-Space-LRF} and the algorithms for  continual releases follows from the intuition given in~\secref{contributions}. 
The basic idea behind {\scshape Spectral-LRF} is to compute matrix $\bU$ and find $\bS$ with certain properties.  
  Proving that our candidate choices for $\bS$ and $\bU$ satisfy the desired properties turns out to be the main technical challenge of this paper. The main complication here is that the spectral norm does not satisfy the Pythagorean theorem, which forms a key ingredient in the proof when the approximation guarantees are in the Frobenius norm~\cite{Upadhyay16}. 
In what follows,  we use the private matrix $\bA$ to give the intuition. In the actual proof, there is an extra Gaussian matrix due to use of the Gaussian mechanism to preserve privacy. We show that if we use the Gaussian mechanism appropriately, the additive error scales optimally with the dimensions of the private matrix $\bA$. 

Let $[\bU]_k [\bSigma]_k [\bV]_k^{\mathsf T}$ be a singular-value decomposition of $[\bA]_k$ and $\mathsf{r}(\cdot)$ denote the rank of the matrix. 
Let $\bPhi$ be a matrix that satisfies~\lemref{mult} for parameters $(\sqrt{\alpha}/k,1/10)$. Then our choice of $\bU$ is an orthonormal basis of the column space of $\bA \bPhi$. This implies that  $\min_{\mathsf{r}(\bX) \leq k} \| \bU \bX - \bA \|_2 \leq \min_{\mathsf{r}(\bX) \leq k} \| \bA \bPhi \bX - \bA \|_2$. We use two optimization problems to prove that $\bU$ is a good approximation to $[\bU]_k$: $ \min_{\bX} \| \bPhi^\sT ([\bV]_k \bX - \bA^\sT) \|_2$ ({\scshape Problem 1}) and  $ \min_{\bX} \| ([\bV]_k \bX - \bA^\sT) \|_2$ ({\scshape Problem 2}). 
 The result follows  by combining~\eqnref{intuition3},~\eqnref{intuition1}, and~\eqnref{intuition2} listed below.

Since {\scshape Problem 1} is a minimization problem over all allowable matrices $\bX$, using~\thmref{FT07} and the singular value decomposition of $\bPhi^\sT [\bV]_k$, we get $\t{\bX} = (\bPhi^\sT [\bV]_k)^\dagger \bPhi^\sT \bA^\sT $. Since the rank of $([\bV]_k (\bPhi^\sT [\bV]_k)^\dagger)$ is at most $ k$,   
 \begin{align}
 \min_{\mathsf{r}(\bX) \leq k} \| \bU \bX - \bA \|_2 &\leq 
  \min_{\mathsf{r}(\bX) \leq k} \| \bA \bPhi \bX - \bA \|_2 \nonumber  \\ 
  &\leq
  \| \bA \bPhi ([\bV]_k (\bPhi^\sT [\bV]_k)^\dagger)^\sT - \bA \|_2 \nonumber \\ 
  & =
 \| [\bV]_k \t{\bX} - \bA^\sT \|_2 . 
 \label{eq:intuition3}
 \end{align} 

We then show that, if $\t{\bX}$ is a solution to {\scshape Problem 1} and $\h{\bX}$ is a solution to {\scshape Problem 2}, then 
\begin{align}
\| [\bV]_k \t{\bX} - \bA^\sT \|_2 \leq (1-\alpha)^{-1} \| ([\bV]_k \h{\bX} - \bA^\sT) \|_2\label{eq:intuition1}.
\end{align}

Now if we pick $\bX=[\bSigma]_k[ \bU]_k^\sT$ in  {\scshape Problem 2}, then
 \begin{align}
 \| [\bV]_k \h{\bX} - \bA^\sT \|_2 &= \min_{\bX} \| ([\bV]_k \bX - \bA^\sT) \|_2 \nonumber \\
 & \leq  \| \bA^\sT - [\bV]_k[\bSigma]_k[ \bU]_k^\sT \|_2 = \Delta_k(\bA). \label{eq:intuition2}
 \end{align}

We sample $\bS$ from a distribution of random projection matrices that satisfies the {\em Johnson-Lindenstrauss lemma}. To prove that our choice of $\bS$ satisfies the required property, we introduce the  optimization problem: 
$ \min_{\bX} \| \bU \bX - \bA \|_2. $
Since $\bX$ can have rank at most the rank of $\bA$, using~\thmref{FT07}, we have that $\h{\bX}=\bU^\sT [\bU \bU^\sT \bA]= \argmin_{\bX} \| \bU \bX - \bA \|_2.$ This implies that $\bU^\sT(\bU \h{\bX}-\bA) = \bU^\sT(\bU\bU^\sT -\I)\bA =0$. We use this fact, the sub-additivity and sub-multiplicativity of the spectral norm, and the fact that $\|\bS \|_2 \approx 1$ if $\bS$ satisfies the Johnson-Lindenstrauss lemma to prove that $\bS$ satisfies the required property.

To preserve the privacy of {\scshape Spectral-$\lrf$}, we note that the private matrix is used twice: when we compute the  matrix $\bU$ and when we compute $\bS \bA$. Therefore, if we add appropriately scaled Gaussian matrices at these stages, then the  privacy of the algorithm follows from the fact that differential privacy is preserved under post-processing.
For the privacy proof of {\scshape Low-Space-LRF}, we use the idea used by Upadhyay~\cite{Upadhyay16} to use both the input and the output perturbation. 

\subsection{Comparison with Previous Works} \junote{Discuss how hard it is to compare with previous works. Limits of streaming and multiplicative approximation. Use Frobenius norm example.}
All the previous private algorithms compute low-rank approximation (either of the matrix or its covariance). Though it is possible to compute the factorization of their outputs, this would incur an extra $O(mn^2)$ running time for the factorization of an $m \times n$ rank-$k$ matrix. Moreover,  they require $O(mn)$ space just to store the  output (Dwork {\it et al.}~\cite{DTTZ14} requires $O(n^2)$, but they work with covariance matrices). Unlike this work, previous works also assume structured private matrix, like matrices with  singular value separation~\cite{HR13,HP14,KT13}. 
Previous works have studied three variants of $\lra$. Dwork {\it et al.}~\cite{DTTZ14} only approximate the right-singular vector while Hardt and Roth~\cite{HR13} (and subsequently, Hardt and Price~\cite{HP14} who improved the result of Hardt and Roth~\cite{HR13} and Upadhyay~\cite{Upadhyay14}) approximates both  right and left singular vectors. %
 Kapralov and Talwar~\cite{KT13} and {Jiang {\it et al.}~\cite{JXZ15}}  considered $(\varepsilon,0)$-differential privacy. On the other hand, our algorithm outputs a factorization with $(\varepsilon,\delta)$-differential privacy.  In this sense, our problem is closely related to Hardt and Price~\cite{HP14}. {We give the technical descriptions and the differences between Problem~\ref{prob:private_factor} and all the previously studied problems in~\appref{related}. } 
Here, we only compare this work with the results of Hardt and Price~\cite{HP14} and Upadhyay~\cite{Upadhyay14}. We give a comparison of the results  in Table~\ref{table:results}. 
Below, we enumerate the key differences between our result and that of Hardt and Price~\cite{HP14}. 

\begin{enumerate} 
	\item Space efficiency. The  algorithms of Hardt and Price~\cite{HP14} and Upadhyay~\cite{Upadhyay14} use  $O(mn)$ space. On the other hand, both of our algorithms use sub-linear space. 
	\item  Approximation error. Both of our bounds improve the additive error of Upadhyay~\cite{Upadhyay14} by a factor of $k^{3/2}\alpha^2 \sqrt{\log(1/\delta)}$. To make a reasonable comparison with Hardt and Price~\cite{HP14}, we consider their result without  coherence assumption.  Hardt and Price~\cite{HP14} incurs  an additive error $\widetilde O(\sigma_1 \sqrt{k(m+n) } \varepsilon^{-1})$, where $\sigma_1$ is the maximum singular value of the input matrix. We can rewrite our results to say that we incur an additive error that depends on $\alpha \sigma_{k+1} + \widetilde O (\sqrt{m+n}/\varepsilon)$. In most real world scenarios, $\sigma_{k+1} \ll \sigma_1$. In other words, we improve the result of Hardt and Price~\cite{HP14}  if $\sigma_1$  is large. Recall that the dependency on $\alpha$ is unavoidable for low-space algorithm (see, Remark~\ref{remark:space}).
	\junote{Talk about multiplicative error as well!} 
	\item Streaming constraints. Hardt and Price~\cite{HP14} is a private version of the iterative algorithm of Halko {\it et al.}~\cite{HMT11}. The algorithm of Upadhyay~\cite{Upadhyay14} is one-pass, but assumes that the matrix is streamed row-wise. A row-wise update is an easier problem (with respect to the space required) compared to the turnstile update model even  in the non-private setting as illustrated by Clarkson and Woodruff~\cite{CW09}.
\end{enumerate}

\begin{table} [t]
{
\small{
 \begin{center}
\begin{tabular}{|c|c|c|}
\hline
	  									& Neighbouring Data 								  & Multiplicative Error \\
						&	Assumptions 			 & Additive Error
  \\ \hline
{\thmref{naive}}						& $\| \bA - {\bA}'\|_F = 1$	& $ \gamma={(1+\alpha)}{(1-\alpha)^{-2}}$  \\
 $(\varepsilon,\delta)$-differential privacy    & No assumption & $ \zeta=\widetilde O  \paren{ \naiveadditive  }$ 
\\ \cline{1-3}
{\thmref{space}}						& $ \bA - {\bA}' = \mathbf{u} \mathbf{v}^{\mathsf T}$	& $ \gamma={(1+\alpha)^2}{(1-\alpha)^{-4}}$  \\
 $(\varepsilon,\delta)$-differential privacy    & No assumption & $ \zeta=\widetilde O \paren{\spaceadditive}$ 
\\ \hline

{Kapralov and Talwar~\cite{KT13}}  &	 $\| \bA\|_2 -  \|{\bA}' \|_2 = 1$ 						& $ \gamma=1$	 	 \\ 
 $(\varepsilon,0)$-differential privacy   & Singular-value separation (SVS) & $ \zeta=c \lambda_1,~\text{where}~\lambda_1 = \Omega \paren{ \frac{nk^3}{\varepsilon c^6} }$
   \\ \cline{1-3}

{Upadhyay~\cite{Upadhyay14}}					& $\bA - {\bA}'=\mathbf{e}_s \mathbf{v}^{\mathsf T}$	& $ \gamma=\poly(k)$	 \\ 
 $(\varepsilon,\delta)$-differential privacy   & No assumption & $\zeta= \widetilde O  \paren{  {k^2 \varepsilon^{-1}\sqrt{ (m+n) }  }}  $
  \\ \cline{1-3}
  
{Hardt and Price~\cite{HP14}}					& $\bA - {\bA}'=\mathbf{e}_s \mathbf{e}_t^{\mathsf T}$	& $ \gamma=1$	 \\ 
  $(\varepsilon,\delta)$-differential privacy   & Incoherence and SVS & $\zeta= \widetilde O  \paren{  \frac{\sigma_1\sqrt{k\mu \log d\log(\log d\sigma_k/(\sigma_{k}-\sigma_{k+1}))}  }{\varepsilon(\sigma_k - \sigma_{k+1})}  }$
   \\ \cline{1-3}

{Hardt and Price~\cite{HP14}}					& $\bA - {\bA}'=\mathbf{e}_s \mathbf{e}_t^{\mathsf T}$	& $ \gamma=1$ \\ 
 $(\varepsilon,\delta)$-differential privacy	& SVS   & $\zeta= \widetilde O  \paren{  \frac{\sigma_1\sqrt{k d \log d  \log(\log(d)\sigma_k/(\sigma_{k}-\sigma_{k+1}))}  }{\varepsilon(\sigma_k - \sigma_{k+1})}  }$
  \\ \hline

{Dwork {\it et al.}~\cite{DTTZ14}}			& $\bA - {\bA}' = \mathbf{e}_s \mathbf{v}^{\mathsf T}$ 	&$ \gamma=1$	 \\ 
 $(\varepsilon,\delta)$-differential privacy  & No assumptions & $ \zeta=\widetilde O  \paren{ {k\sqrt{n} }/{\varepsilon} }$
 \\ \cline{1-3}
 
{Jiang {\it et al.}~\cite{JXZ15}} 					& $\bA - {\bA}' = \mathbf{e}_s \mathbf{e}_t^{\mathsf T}$ 		& $ \gamma=1$	 \\ 
  $(\varepsilon,0)$-differential privacy  & No assumption & $ \zeta=O\paren{ {n \varepsilon^{-1} \log n}}$
 \\ \hline
\end{tabular}
\caption{Comparison of Models for Differentially Private $k$-Rank Approximation ($\beta=O(1)$, $\mathbf{u}$ and $\mathbf{v}$ are unit vectors, $\mathbf{e}_s$ is the $s$-th standard basis, and $d=m+n$).} \label{table:results}
\label{table}
 \end{center}
}}
\end{table}

\subsection{Related Works}
Low-rank approximation of large data-matrices  has received a lot of attention in the recent past in the private as well as the non-private setting.  In what follows, we give a brief exposition of those that are most relevant to this work.

In the private setting, previous works have either used random projection~\cite{CW09,KN14,Sarlos06} or random sampling (at a cost of  a small additive error) to give low-rank approximation~\cite{AM07,DKM06,DM05,DRVW06,FKV04,PTRV98,RV07}. Many of the latter algorithms were improved independently by Deshpande and Vempala~\cite{DV06} and Sarlos~\cite{Sarlos06}.  Subsequent works~\cite{CW13, MZ11,MM13,NDT09} achieved a run-time that depends linearly on the input sparsity of the matrix. 
In a series of works, Clarkson and Woodruff~\cite{CW09,CW13} showed space lower bounds and almost matching space algorithms. 
Recently, Boutsidis {\it et al.}~\cite{BWZ16} gave the first space-optimal algorithm for low-rank approximation, but they do not optimize for run-time.

In the private setting, $\lra$ has been studied under a privacy guarantee called differential privacy. Differential privacy was  introduced by Dwork {\it et al.}~\cite{DMNS06}. The Gaussian variant of this basic sanitizer was proven to preserve differential privacy by ~\cite{DKMMN06} in a follow-up work. 
Since then, many algorithms for preserving differential privacy have been proposed in the literature~\cite{DR14}.  All these mechanisms have a common theme: they perturb the output before responding to  queries. Recently, Blocki {\it et al.}~\cite{BBDS12} and Upadhyay~\cite{Upadhyay13} took a complementary approach. They perturb the input reversibly and then perform a random projection of the perturbed matrix.

Blum  {\it et al.}~\cite{BDMN05} first studied the problem of differentially private $\lra$ in the Frobenius norm. This was improved by Hardt and Roth~\cite{HR12}  under the low coherence assumption. Upadhyay~\cite{Upadhyay14} later made it a single-pass. Differentially-private $\lra$ has been studied in the spectral norm as well by many works~\cite{CSS12,KT13,HR13,HP14}.
Kapralov and Talwar~\cite{KT13} and Chaudhary  {\it et al.}~\cite{CSS12} studied the spectral $\lra$ of a matrix by giving a matching upper and lower bounds for privately computing the top $k$ eigenvectors of a matrix with pure differential privacy (i.e., $\delta=0$). In subsequent works Hardt and Roth~\cite{HR13} and Hardt and Price~\cite{HP14} improved the approximation guarantee with respect to the spectral norm by using {\em robust private subspace iteration} algorithm. 
Recently, Dwork  {\it et al.}~\cite{DTTZ14} gave a tighter analysis of Blum {\it et al.}~\cite{BDMN05} to give an optimal approximation to the right singular space, i.e., they gave a $\lra$ for the covariance matrix.  

The literature on non-private algorithms that makes few passes over the input matrix is so  extensive that we cannot hope to cover it in any detail here. In the private setting, Dwork  {\it et al.}~\cite{DNPR10} first considered streaming algorithms with privacy under the model of {\em pan-privacy}, where the internal state is known to the adversary. They gave private analogues of known sampling based streaming algorithms to answer various counting tasks, like estimating  distinct elements, cropped means, number of heavy hitters, and frequency counts. 
This was followed by results on online private learning~\cite{DTTZ14,JKT12,TS13}.


\section{Notations and Previous Known Results Used in This Paper}\label{sec:prelims}
 We let $\N$ to denote the set of natural numbers. We use bold-face capital letters to denote matrices and bold-face small letters to denote vectors. We denote by $\mathbf{0}^{m \times n}$ the all-zero $m \times n$ matrix and by $\I_n$ the $n \times n$ identity matrix. For a matrix $\bA$, we denote its best $k$-rank approximation by $[\bA]_k$, its Frobenius norm by $\| \bA \|_F$, and its spectral norm by $\| \bA \|_2$. The {\em singular-value decomposition} (SVD) of an $m \times n$ rank-$r$ matrix $\bA$ is a decomposition of $\bA$ as a product of three matrices, $\bA = \bU \bSigma \bV^{\mathsf T}$ such that $\bU \in \R^{m \times r}$ and $\bV \in \R^{n \times r}$ have orthonormal columns and $\bSigma \in \R^{r \times r}$ is a diagonal matrix with singular values of $\bA$ on its diagonal. For a matrix $\bA$, we use the symbol $\mathsf{r}(\bA)$ to denote its {\em rank} and $\det(\bA)$ to denote its {\em determinant}. The {\em Moore-Penrose pseudo-inverse} of a matrix $\bA =  \bU \bSigma \bV^{\mathsf T}$ is denoted by $\bA^\dagger$ and has a SVD $\bA^\dagger =  \bV \bSigma^\dagger \bU^{\mathsf T}$, where $\bSigma^\dagger$ consists of inverses of only non-zero singular values of $\bA$. 
Given a random variable $x$, we denote by $\cN(\mu, \rho^2)$ the fact that  $x$ has a normal Gaussian distribution with mean $\mu$ and variance $\rho^2$. 

Let $\alpha, \beta >0$. A distribution $\cD$ over ${p \times n}$ random matrices satisfies {\em $(\alpha,\beta)$-Johnson-Lindenstrauss property} {(\em {\sf JLP})} if, for   any unit vector $\bx \in \R^n$, we have $(1- \alpha) \leq \|\bPhi \bx \|_2^2  \leq (1 + \alpha)$ with probability $1-\beta$ over $\bPhi \sim \cD$.  A distribution $\cD$ of $p\times m$ matrices satisfies {\em $(\alpha,\beta)$-subspace embedding} for a matrix $\bA \in \R^{m \times n}$ if, for all $\bx \in \R^n$, with probability $1-\beta$ over $\bPhi \sim \cD$,  $(1-\alpha) \|  \bA \bx \|_2 \leq \| \bPhi \bA \bx \|_2 \leq (1+\alpha) \|  \bA \bx \|_2.$

 \begin{lemma} \label{lem:mult} (Clarkson and Woodruff~\cite{CW13})
Let $\bB_1$ and $\bB_2$ be arbitrary matrices with $m$ rows such that $\bB_1$ has rank-$r$. Let $\cD$ be a distribution over $p \times m$ random matrices that satisfies the $(\alpha,\beta)$-subspace embedding for $\bB_1$. 
 Then there exists a $p=\Theta(\alpha^{-2})$ such that, with probability at least $1-\beta$ over $\Phi \sim \cD$, 
 \begin{align*} {\| \bB_1^{\mathsf T} \bPhi^{\mathsf T} \bPhi \bB_2 - \bB_1^{\mathsf T} \bB_2 \|_F^2 \leq {\alpha^2} \| \bB_1 \|_F^2 \| \bB_2 \|_F^2 }.\end{align*}
\end{lemma}
The tuple $(\alpha,\beta)$ is called the {\em error parameters.}

\section{Differentially Private Algorithm for Spectral {\scshape LRF} in the Turnstile Model} \label{sec:naive}
In this section, we present our space-efficient private algorithm for  $\lrf$ under $\mathsf{Priv}_1$ with respect to the spectral norm in the turnstile update model. Our privacy level, $\priv_1$, is stronger than that considered by previous $(\varepsilon,\delta)$-differentially private algorithms~\cite{DTTZ14,HR13,HP14}. Our algorithm formalizes the ideas mentioned in~\secref{overview}. In particular, we show that we can compute a set of $k$-orthonormal basis vectors that approximates the top-$k$ left singular space of the private matrix in the turnstile update model.
The details of our algorithm {\scshape Spectral-LRF} is presented in~\figref{private}. We prove the following bound about   {\scshape Spectral-LRF}.


\begin{figure} [t]
\begin{center} 
\fbox
{
\begin{minipage}[l]{5.5in}
\small{
\medskip \noindent \textbf{Initiialization.} 
	Let $\eta:=\max\set{k^2,1/\alpha}$, $\rho:={\sqrt{(1+\alpha) \ln(1/\delta)}}/{\varepsilon}$,  $t=O(\eta \alpha^{-1} \log(k/\alpha) \log (1/\delta))$, and $v=O(\eta \alpha^{-3} \log (t/\alpha)  \log(1/\delta))$. 
Sample $\bN_1 \sim \cN(0,\rho^2)^{m \times t}$ and $\bN_2 \sim \cN(0,\rho^2)^{v \times n}$ 
	Let $\bPhi \in \R^{n \times t}$ be such that $\bPhi \sim \cD_R$ satisfies the statement of~\lemref{mult} for the error parameters $(\sqrt{\alpha}/k,\delta)$.
	Let  $\bS \in \R^{v \times m}$ such that $\bS \sim \cD_A$ satisfies $(\alpha^2,\delta)$-{\sf JLP}. Initialize  all ${m \times t}$ zero matrix $\h{\bY}$ and an all zero ${v \times n}$ matrix $\h{\bZ}$. Publish $\bS$ and $\bPhi$. $\bN_1$ and $\bN_2$ are private matrices.

\medskip
\medskip \noindent \textbf{Update rule and computing the factorization.}  Suppose at time $\tau$, the stream is $(i_\tau,j_\tau,s_{\tau})$, where $(i_\tau,j_\tau) \in [m] \times [n]$. Let $\bA_\tau$ be a matrix with the only non-zero entry $s_\tau$ in  the position $(i_\tau,j_\tau)$. Update the matrices by the following rule: $\h{\bY} \leftarrow \h{\bY} + \bA_\tau \bPhi$ and  $\h{\bZ} \leftarrow \h{\bZ} +  \bS\bA_\tau$. 

\medskip
 Once the matrix is streamed, we follow the following steps.
\begin{enumerate}
	\item 
	 Compute the singular value decomposition $\widetilde{\bU} \widetilde{\bSigma} \widetilde{\bV}^{\mathsf T}$ of $\mathbf{S} \bU \in \R^{v \times t}$, where $\bU \in \R^{m \times t}$ is a matrix whose columns are an orthonormal basis  for the column space of $(\h{\bY} + \bN_1)$. 
	\item Compute the singular value decomposition of $ \widetilde{\bV} \widetilde{\bSigma}^{\dagger} \widetilde{\bU}^{\mathsf T} [\widetilde{\bU} \widetilde{\bU}^{\mathsf T}(\h{\bZ}+\bN_2)]_k \in \R^{t \times n}$. 
	Let it be $\bU' \bSigma' \bV'^{\mathsf T}$. 
	\item Output $\bU_k:=\bU \bU'$, $\bSigma_k:=\bSigma'$ and $\bV_k:=\bV'$. Let $\mathbf{M}_k = \bU_k \bSigma_k \bV_k^{\mathsf T}.$
\end{enumerate}
}\end{minipage}
} \caption{Differentially private Low-rank Factorization ({\scshape Spectal-LRF})} \label{fig:private}
\end{center}
\end{figure}

\begin{theorem} \label{thm:naive}
\naivetheorem
 \end{theorem}
When $m <n$, we can get similar bounds with the roles of $m$ and $n$ reversed and by using $\bA^\sT$ in~\figref{private}.
The key technical point in the proof is to prove that $\bU$ is indeed a good approximation to the top-$k$ left singular vectors of $\bA$ and $\bS$ has the property that, simultaneously for all matrices $\bX$ of appropriate dimensions, $\| \bS(\bU \bX - \bA)\|_2 \leq \frac{(1+\alpha)}{(1-\alpha)} \| (\bU \bX - \bA)\|_2$.

\begin{proof} [Proof of~\thmref{naive}.]
The space complexity of the algorithm is straightforward from the choice of $v$ and $t$. The proof of correctness consists of two main lemmas (\lemref{complete} and~\lemref{naiveN_1}). In~\lemref{complete}, we bound $\| \mathbf{M}_k - \bA \|_2 $ by $\|\bA - [\bA]_k\|_2$ and two additive terms. In~\lemref{naiveN_1}, we bound the two additive terms. 

For all our correctness proofs, we make a standard assumption that $\delta \ll 1/100$.  Let $\bY = \h{\bY}+\bN_1$, $\bZ=\h{\bZ}+\bN_2$, and  $\bB=\bA+\bS^\dagger  \bN_2 $ for $\h{\bY}, \h{\bZ}, \bS, \bN_1, \bN_2$ as in~\figref{private}.
\begin{lemma} \label{lem:complete}
Let $\bU_k,\bSigma_k,\bV_k$  be the  output of the algorithm  {\scshape Spectral-LRF} presented in~\figref{private} such that $\mathbf{M}_k:=\bU_k\bSigma_k\bV_k^\sT$. Then with probability $1-O(\delta)$ over $\bPhi \sim \cD_R$ and $\bS \sim \cD_A$,
\begin{align*} \| \mathbf{M}_k - \bA \|_2 & \leq \frac{(1+\alpha)}{(1-\alpha)^2}  \|\bA - [\bA]_k\|_2 + \frac{2 \| \bN_2\bS^\dagger  \|_2}{1-\alpha}  + \frac{1+\alpha}{1-\alpha}  \| \bN_1([\bV]_k^\sT \bPhi)^\dagger [\bV]_k^\sT \|_2.   \end{align*}
\end{lemma}
 
\begin{proof}
We prove~\lemref{complete} by proving two seperate claims.
~\claimref{firstcomplete} shows that $\bU$ is a ``faithful" representation of the top-$k$ left singular vectors of $[\bA]_k$ and~\claimref{thirdcomplete} shows that $\bS$ satisfies the required properties with the choice of error parameters.

\begin{claim} \label{claim:firstcomplete}
	Let $\bA$ be the input matrix. 
	Let $\bPhi \sim \cD_R$ be a random matrix that satisfies~\lemref{mult} with error parameters $(\sqrt{\alpha}/k,\delta)$ with respect to a rank-$k$ matrix $[\bV]_k$.  
	Then with probability $1-\delta$ over $\bPhi \sim \cD_R$,
	\begin{align*} \min_{\bX, \mathsf{r}(\bX) \leq k} \| \bU \bX - \bB \|_2 &\leq \frac{(1+\alpha)}{(1-\alpha^2)} \Delta_k(\bA) 
	+ \| \bN_1([\bV]_k^\sT \bPhi)^\dagger [\bV]_k^\sT \|_2 + \| \bS^\dagger \bN_2   \|_2.  \end{align*}
\end{claim}
Before, we prove~\claimref{firstcomplete}, we prove some auxiliary results. Let us denote by $[\bU]_k [\bSigma]_k [\bV]_k^{\mathsf T}$  the SVD of $[\bA]_k$. We first note  that $\bU$ is an orthonormal basis for the column space of $\bY$, i.e., 
	\begin{align} 
		\min_{\bX, \mathsf{r}(\bX) \leq k} \| \bU \bX  - \bB \|_2 \leq \min_{\bX, \mathsf{r}(\bX) \leq k} \| \bY \bX  - \bB \|_2. \label{eq:basis2}  
	\end{align}
	We also note that the choice of $t$ allows $\bPhi$ to satisfy \lemref{mult} with error parameters $(\sqrt{\alpha}/k,\delta)$~\cite{CW13}. 
Therefore, in order to show that $\bU$ approximates the top-$k$ left singular vectors of the matrix $\bA$, we need to find an appropriate matrix $\bX$ such that $\| \bY \bX - \bA \|_2$ is bounded by a multiplicative factor of $ \|\bA - [\bA]_k \|_2$. Towards this goal,  consider the following two optimization problems:
\begin{align} \min_{\bX} \| \bPhi^\sT ([\bV]_k \bX - \bA^\sT) \|_2 \quad \text{and} \quad \min_{\bX} \| [\bV]_k \bX - \bA^\sT \|_2. \label{eq:regressionproblem}
\end{align}

with 
	\begin{align*} & \t{\bX}:=\argmin_{\bX} \| \bPhi^\sT([\bV]_k \bX - \bA^\sT) \|_2 \\  & \h{\bX}:=\argmin_{\bX} \| [\bV]_k \bX - \bA^\sT \|_2  \end{align*}
as one of the solutions. We first give the sketch of the proof.
Our goal is to show that 
\begin{align}
\| [\bV]_k \t{\bX} - \bA^\sT \|_2 \leq (1-\alpha)^{-1} \| [\bV]_k \h{\bX} - \bA^\sT \|_2 . \label{eq:aim}
\end{align}
 This would give us~\claimref{firstcomplete} after substituting the value of $\bB$ and $\bY$ because of the following argument.  
 
 Since $\h{\bX}$ minimizes {\scshape Problem 2}, $\bX=[\bSigma]_k[ \bU]_k^\sT$ would only increase the value of $\| ([\bV]_k \bX - \bA^\sT) \|_2$. That is, 
 \begin{align}
 \| [\bV]_k \h{\bX} - \bA^\sT \|_2& = \min_{\bX} \| ([\bV]_k \bX - \bA^\sT) \|_2 \nonumber  \\ & \leq  \| \bA^\sT - [\bV]_k[\bSigma]_k[ \bU]_k^\sT \|_2 \nonumber \\ &= \Delta_k(\bA). \label{eq:intuition21}
 \end{align} 

Now since {\scshape Problem 1} is a minimization problem over all possible matrices $\bX$, using~\thmref{FT07} and the singular value decomposition of $\bPhi^\sT [\bV]_k$, we get $\t{\bX} = (\bPhi^\sT [\bV]_k)^\dagger \bPhi^\sT \bA^\sT $. Furthermore, $\mathsf{r}([\bV]_k (\bPhi^\sT [\bV]_k)^\dagger) \leq k$. This implies that
 \begin{align}
 \min_{\mathsf{r}(\bX) \leq k} \| \bU \bX - \bA \|_2 &\leq 
  \min_{\mathsf{r}(\bX) \leq k} \| \bA \bPhi \bX - \bA \|_2 \nonumber \\ &\leq
  \| \bA \bPhi ([\bV]_k (\bPhi^\sT [\bV]_k)^\dagger)^\sT - \bA \|_2 \nonumber \\ &=
 \| [\bV]_k \t{\bX} - \bA^\sT \|_2 . 
 \label{eq:intuition31}
 \end{align} 

We get the desired result by combining~\eqnref{aim}, \eqnref{intuition21}, \eqnref{intuition31}, and substituting the value of $\bB= \bA + \bS^\dagger \bN_2$ and $\bY = \h{\bA} + \bN_1$. Therefore, the key to our proof is to prove~\eqnref{aim}. 

\medskip \noindent \textbf{Proving~\eqnref{aim}.}
We now prove~\eqnref{aim}. Using the sub-additivity of the spectral norm, we have
\begin{align*} \| \bA^\sT - [\bV]_k \t{\bX} \|_2  \leq \| \bA^\sT  - [\bV]_k \h{\bX} \|_2 + \| [\bV]_k (\t{\bX} - \h{\bX}) \|_2. \end{align*}

The first term in the above expression has the form that we desire. Therefore, towards the goal we set forth, we need to bound $\| [\bV]_k (\t{\bX} - \h{\bX}) \|_2$. 
Let $\bC= [\bV]_k^{\sT} [\bV]_k (\t{\bX} - \h{\bX})$ and $\bD = [\bV]_k(\t{\bX} - \h{\bX})$. It is easy to see that $\bC^\sT \bC = \bD^\sT \bD$; therefore $\|\bD\|_2=\| \bC \|_2$.  This implies that instead of $\bD$ if we bound $\bC$, then we are done. We first make an observation that $[\bV]_k^\sT \bPhi \bPhi^\sT  (\bA^\sT - [\bV]_k \t{\bX})= 0$ because $\I - \bPhi^\sT [\bV]_k (\bPhi^\sT [\bV]_k)^\dagger$ is a projection on to the space orthogonal to $[\bV]_k^\sT \bPhi$	 and $[\bV]_k^\sT \bPhi \bPhi^\sT  (\bA^\sT - [\bV]_k \t{\bX})$ equals 
	\begin{align}
		[\bV]_k^\sT \bPhi  (\bPhi^\sT \bA^\sT - \bPhi^\sT [\bV]_k (\bPhi^\sT [\bV]_k)^\dagger \bPhi^\sT \bA^\sT)  = 0 \label{eq:orthogonal}.
	\end{align}

 Now, we return to the proof of~\eqnref{aim}. The sub-additivity and sub-multiplicativity of the spectral norm  gives us the following:
\begin{align*}
	\| \bC \|_2 &= \| [\bV]_k^\sT \bPhi \bPhi^\sT  [\bV]_k \bC - [\bV]_k^\sT \bPhi\bPhi^\sT  [\bV]_k \bC + \bC \|_2 \\
			&\leq \| [\bV]_k^\sT \bPhi\bPhi^\sT  [\bV]_k \bC  \|_2 + \| ([\bV]_k^\sT\bPhi \bPhi^\sT  [\bV]_k - \I) \bC \|_2  \\
			&\leq \| [\bV]_k^\sT \bPhi\bPhi^\sT  [\bV]_k \bC  \|_2 + \| (\bPhi \bPhi^\sT   - \I)\|_2 \| \bC \|_2.   
\end{align*}
			{The choice of the dimension of $\bPhi$ allows us to use the result of Clarkson and Woodruff~\cite{CW09} (\lemref{CW09}).  }
\begin{align*}
	\|\bC\|_2 &\leq (1 - \alpha^2)^{-1}	\| [\bV]_k^\sT \bPhi\bPhi^\sT  [\bV]_k \bC  \|_2 \\
			& \leq (1 - \alpha^2)^{-1}	\| [\bV]_k^\sT \bPhi\bPhi^\sT  [\bV]_k \bC  \|_F\\
			&= (1 - \alpha^2)^{-1}\| [\bV]_k^\sT \bPhi\bPhi^\sT  [\bV]_k( \t{\bX} - \h{\bX} ) 
			+ [\bV]_k^\sT \bPhi \bPhi^\sT  (\bA^\sT - [\bV]_k \t{\bX}) \|_F  \\
			&= (1 - \alpha^2)^{-1} \| [\bV]_k^\sT \bPhi\bPhi^\sT  (\bA^\sT - [\bV]_k \h{\bX}) \|_F \\
			&\leq (1 - \alpha^2)^{-1} \alpha   \| (\bA^\sT - [\bV]_k \h{\bX}) \|_F /\sqrt{k}.  \\
			&\leq (1 - \alpha^2)^{-1} \alpha   \| (\bA^\sT - [\bV]_k \h{\bX}) \|_2,
\end{align*}
the equality follows from~\eqnref{orthogonal}), the second and last inequality follows from Fact~\ref{fact:normequivalence}. 

In the above inequalities, we  were able to  apply~\lemref{mult} because $[\bV]_k$ is a rank-$k$ matrix and $\bPhi$ satisfies $(\sqrt{\alpha}/k,\delta)$-subspace embedding with respect to $[\bV]_k$. Now, from the sub-additivity of the spectral norm,
\begin{align}
	\| \bA^\sT - [\bV]_k \t{\bX} \|_2  &\leq \| \bA^\sT  - [\bV]_k \h{\bX} \|_2 + \| [\bV]_k (\t{\bX} - \h{\bX}) \|_2 \nonumber  \\
	&\quad \leq (1+\alpha-\alpha^2)(1 - \alpha^2)^{-1} \| \bA^\sT - [\bV]_k \h{\bX} \|_2\nonumber \\
	&\quad \leq (1+\alpha)(1 - \alpha^2)^{-1} \| \bA^\sT - [\bV]_k \h{\bX} \|_2 \nonumber \\
	&\quad = (1 - \alpha)^{-1} \| \bA^\sT - [\bV]_k \h{\bX} \|_2.	\label{eq:generalizedreg}	
\end{align}
This gives us the desired expression stated in~\eqnref{aim}.

\begin{proof}
[{Proof of~\claimref{firstcomplete}.}] We now complete the proof of~\claimref{firstcomplete}. We follow the sketch of the proof given earlier. By the definition of $\h{\bX}$, we have 	
	\begin{align*} 
		&\| \bA^\sT - [\bV]_k \h{\bX} \|_2 = \min_{\bX} \| [\bV]_k \bX - \bA^{\mathsf T} \|_2 
		\leq \| [\bV]_k  [\bSigma]_k[\bU]_k^{\mathsf T} - \bA^{\mathsf T} \|_2  = \Delta_k(\bA)
	\end{align*}
	by setting $\bX=   [\bSigma]_k[\bU]_k^{\mathsf T}$.
	   Using~\eqnref{aim}, substituting the value of $\t{\bX}$, taking the transpose, and the fact that the spectral norm is preserved under transpose, we have with probability $1-\delta$ over $\bPhi \sim \cD_R$,  
	\begin{align} 
		\| \bA \bPhi ([\bV]_k^\sT \bPhi)^{\dagger} [\bV]_k^\sT - \bA \|_2 \leq (1 - \alpha)^{-1} \Delta_k(\bA). \label{eq:transpose2}  
	\end{align}
	
	Moreover, since $([\bV]_k \bPhi)^{\dagger} [\bV]_k $ has rank at most $k$ and $\bY=\bA \bPhi + \bN_1$, with probability $1-\delta$ over $\bPhi \sim \cD_R$,  
	\begin{align} 
		\min_{\bX, \mathsf{r}(\bX) \leq k} \| (\bY  \bX  - \bB \|_2 
		& \leq  \| \bY ([\bV]_k^\sT \bPhi)^{\dagger}[\bV]_k^\sT  - \bB \|_2 \nonumber \\
		& \quad =  \| \bA \bPhi ([\bV]_k^\sT \bPhi)^{\dagger} [\bV]_k^\sT + \bN_1([\bV]_k^\sT \bPhi)^\dagger [\bV]_k^\sT - \bB \|_2 \nonumber \\
		& \quad =  \| \bA \bPhi ([\bV]_k^\sT \bPhi)^{\dagger} [\bV]_k^\sT 
		+ \bN_1([\bV]_k^\sT \bPhi)^\dagger [\bV]_k^\sT - \bA - \bS^\dagger \bN_2   \|_2 \nonumber \\
		&\quad  \leq  \| \bA \bPhi ([\bV]_k^\sT \bPhi)^{\dagger} [\bV]_k^\sT - \bA \|_2 \nonumber \\
		&\qquad + \| \bN_1([\bV]_k^\sT \bPhi)^\dagger [\bV]_k^\sT \|_2 + \| \bS^\dagger \bN_2   \|_2\label{eq:set2}  
	\end{align}
	
	Combining~\eqnref{transpose2} and~\eqnref{set2}, we have with probability $1-\delta$ over $\bPhi \sim \cD_R$, 
	\begin{align} \min_{\bX, \mathsf{r}(\bX) \leq k} \| \bY \bX - \bB \|_2 & \leq (1 - \alpha)^{-1} \Delta_k(\bA) 
	+ \| \bN_1 ([\bV]_k^\sT \bPhi)^\dagger [\bV]_k^\sT \|_2 + \| \bS^\dagger \bN_2   \|_2. \label{eq:step12} \end{align}
	
	The above expression relates $\min_{\bX, \mathsf{r}(\bX) \leq k} \| \bY \bX  - \bB \| $ with $(1-\alpha)^{-1} \Delta_k(\bA)$, $ \| \bN_1([\bV]_k^\sT \bPhi)^\dagger [\bV]_k^\sT \|_2$,  and $\| \bS^\dagger \bN_2 \|_2$. 
	Combining~\eqnref{step12} and~\eqnref{basis2}, we have with probability $1-\delta$ over $\bPhi \sim \cD_R$, 
	\begin{align} 
	\min_{\bX, \mathsf{r}(\bX) \leq k} \| \bU \bX  - \bB \|_2 & \leq (1 - \alpha)^{-1}\Delta_k(\bA) 
	+  \| \bN_1([\bV]_k^\sT \bPhi)^\dagger [\bV]_k^\sT \|_2 + \| \bS^\dagger \bN_2   \|_2. \label{eq:step22}
	\end{align}
This completes the proof of the Claim~\ref{claim:firstcomplete}.
\end{proof}	

Note that the terms $ \| \bN_1([\bV]_k^\sT \bPhi)^\dagger [\bV]_k^\sT \|_2$  and $\| \bN_2\bS^\dagger  \|_2$ are due to the addition of Gaussian matrices to preserve privacy. We later bound them in~\lemref{naiveN_1}. We next bound that multiplying $\bS$ from the right does not change our bound by a lot. More concretely, we prove the following claim. 
\begin{claim} \label{claim:thirdcomplete}
Let $\bU, \bB, \bA, \bS, \bN_1$, and $\bN_2$ be as  above, and let $\widetilde{\bX}_k= \argmin_{\bX, \mathsf{r}(\bX)=k} \| \bS(\bU \bX -\bB) \|_2$. Let $\cD_A$ be a distribution that satisfies  $(\alpha^2,\delta)$-{\sf JLP}.  Then  with probability $1-3\delta$ over $\bS \sim \cD_A$ and $\bPhi \sim \cD_R$, we have
\begin{align}
	\frac{1-\alpha}{1+ \alpha}{\|( \bU \widetilde{\bX}_k  - \bB) \|_2} & \leq  \frac{1}{(1-\alpha)}  \Delta_k(\bA)
	+ \| \bN_2\bS^\dagger  \|_2  + \| \bN_1([\bV]_k^\sT \bPhi)^\dagger [\bV]_k^\sT \|_2. 
\end{align}
\end{claim}
\begin{proof}	First note that the choice of $v$ allows $\bS$ to satisfy $({\alpha}^2,\delta)$-{\sf JLP}~\cite{CW13}. 
Let $\h{\bX}_k =  \argmin_{\bX, \mathsf{r}(\bX) \leq k} \| (\bU \bX - \bB) \|_2.$ We want to bound $\| \bU \h{\bX}_k - \bB \|_2$ in terms of $\| \bU \t{\bX}_k - \bB\|_2$. 
We give an even stronger result that states that, simultaneously for all $\bX$,  $\| \bS (\bU \bX - \bB) \|_2$ is bounded by a multiplicative factor of $\|  \bU \bX - \bB\|_2$.

For this consider the following optimization problem:
$ \min_{\bX} \| \bU \bX - \bB \|_2. $
Since $\bX$ can have rank at most the rank of $\bB$, using~\thmref{FT07}, we have that $$\h{\bX}=\bU^\sT [\bU \bU^\sT \bB]=\bU^\sT \bB$$ is a solution to $ \min_{\bX} \| \bU \bX - \bB \|_2. $ This implies that 
\begin{align} 
\bU^\sT(\bU \h{\bX}-\bB) = \bU^\sT(\bU\bU^\sT -\I)\bB =0. \label{eq:normal}
\end{align}

 Therefore, using sub-additivity and sub-multiplicativity of the spectral norm and $(\alpha^2,\delta)$-{\sf JLP} of $\cD_A$, we have with probability $1-\delta$,
 \begin{align*}
 	 \| \bS (\bU \bX - \bB) \|_2^2 &= \| \bS \bU (\bX - \h{\bX}) + \bS(\bU \h{\bX}-\bB) \|_2^2 \\
		&\leq 2(\| \bS \bU (\bX - \h{\bX}) \|_2^2 + \|\bS(\bU \h{\bX}-\bB) \|_2^2) \\
		&\leq 2(1+\alpha^2)^2 \paren{ \|  \bU (\bX - \h{\bX}) \|_2^2 + \|(\bU \h{\bX}-\bB) \|_2^2 } \\
		&\leq 2(1+\alpha^2)^2 \paren{ \|  \bU (\bX - \h{\bX}) \|_2^2 + \|(\bU {\bX}-\bB) \|_2^2 }  \\
		&= 2(1+\alpha^2)^2 \paren{ \| \bU^\sT \bU (\bX - \h{\bX})} \\ &\quad +2(1+\alpha^2)^2 \paren{  \bU^\sT (\bU \h{\bX} - \bB)\|^2_2 + \|(\bU {\bX}-\bB) \|_2^2}   \\ 
		& \leq (4 + 12 \alpha^2)  \|(\bU {\bX}-\bB) \|_2^2,
	\end{align*}
	where the second equality follows from~\eqnref{normal}, the second inequality follows from~\lemref{sarlos},  the third inequality follows from the definition of $\h{\bX}$, and the last inequality follows because $\alpha <1$.
This in particular implies that
\begin{align*}	\| \bS (\bU \bX - \bB) \|_2^2 &- 4   \|(\bU {\bX}-\bB) \|_2^2 \leq 12 \alpha^2  \|(\bU {\bX}-\bB) \|_2^2.
 \end{align*}

{Taking square root and the fact that $$| a - b | \leq \sqrt{(a+b)(a-b)} = \sqrt{a^2-b^2}$$ for positive integers $a$ and $b$, this gives the inequality}
\begin{align*}	\left| \| \bS (\bU \bX - \bB) \|_2 -  2 \|(\bU {\bX}-\bB) \|_2 \right| \leq 2 \sqrt{3}\alpha  \|(\bU {\bX}-\bB) \|_2 \end{align*}

 By rescaling the value of $\alpha$, we have the following:
 \begin{align}
 	2(1-\alpha)  \|  (\bU \bX - \bB) \|_2 &\leq  \| \bS (\bU \bX - \bB) \|_2  \nonumber \\ &\leq 2(1+\alpha)  \|  (\bU \bX - \bB) \|_2 \label{eq:S}
\end{align}
{We now return to the proof of~\claimref{thirdcomplete}. Using the definitions of $\t{\bX}_k$ and $\h{\bX}_k$ along with~\eqnref{S}, we have the following set of inequalities:	}
\begin{align}
	\min_{\bX, \mathsf{r}(\bX)=k} \| \bU \bX  -\bB \|_2 &= \| \bU \widehat{\bX}_k  - \bB \|_2 \nonumber \\ 
			& \geq (2 + 2\alpha)^{-1} \| \bS(\bU  \widehat{\bX}_k  - \bB )  \|_2\nonumber \\
			& \geq (2 + 2\alpha)^{-1}  \min_{\bX, \mathsf{r}(\bX)\leq k} \| \bS(\bU {\bX} - \bB )  \|_2 \nonumber \\
			& = (2 + 2\alpha)^{-1}  \| \bS(\bU \widetilde{\bX}_k - \bB  )  \|_2 \nonumber \\
			&\geq  \frac{1-\alpha}{1+\alpha}  \| (\bU \widetilde{\bX}_k  - \bB  ) \|_2 
			.  \label{eq:S2}  
	\end{align}
	
	Combining~\eqnref{S2} with~\eqnref{step22},  we have with probability $1-3\delta$ over $\bPhi \sim \cD_R$ and $\bS \sim \cD_A$, 
	\begin{align}
	\frac{1-\alpha}{1+\alpha}  \|( \bU \widetilde{\bX}_k - \bB) \|_2  &\leq  {(1-\alpha)}^{-1}  \Delta_k(\bA) 
	+  \| \bS^\dagger \bN_2   \|_2  + \| \bN_1([\bV]_k^\sT \bPhi)^\dagger [\bV]_k^\sT \|_2. \label{eq:step32} 
	\end{align}

This completes the proof of Claim~\ref{claim:thirdcomplete}.
\end{proof}


We can now complete the proof of~\lemref{complete}. Since $\bB= \bA + \bS^\dagger \bN_2 $, using sub-additivity of norm, 
 \begin{align*}
 	   \| \bU \widetilde{\bX}_k   - \bA \|_2 - \| \bS^\dagger \bN_2   \|_2
	  		\leq  \| \bU \widetilde{\bX}_k  - \bB \|_2 
			\leq \frac{(1+\alpha)}{(1-\alpha)^2}  \|\bA - [\bA]_k\|_2 
			+ \frac{1+\alpha}{1-\alpha}  (\| \bS^\dagger \bN_2   \|_2  + \| \bN_1([\bV]_k^\sT \bPhi)^\dagger [\bV]_k^\sT \|_2).
\end{align*}
{This implies that} 
\begin{align*} 	
&\| \bU \widetilde{\bX}_k  - \bA \|_2 	\leq \frac{(1+\alpha)}{(1-\alpha)^2}  \|\bA - [\bA]_k\|_2 
+ \frac{2 \| \bS^\dagger \bN_2  \|_2}{1-\alpha}  + \frac{1+\alpha}{1-\alpha}  \| \bN_1([\bV]_k^\sT \bPhi)^\dagger [\bV]_k^\sT \|_2.  
  \end{align*}

Setting $\mathbf{L} := \bS \bU$, $\mathbf{R}=\I$, and $\mathbf{O} := \bZ$~\thmref{FT07}  gives that $\widetilde{\bX}= \widetilde{\bV} \widetilde{\bSigma}^{\dagger} \widetilde{\bU}^{\mathsf T} [\widetilde{\bU} \widetilde{\bU}^{\mathsf T}\bZ]_k$. 
 Using this value of $\t{\bX}$ completes the proof of~\lemref{complete}. 
\end{proof}
 
Now all that remains to complete the proof of correctness of~\thmref{naive} are the  two expressions in the additive term: $\| \bN_1([\bV]_k^\sT \bPhi)^\dagger [\bV]_k^\sT \|_2$ and $\| \bS^\dagger \bN_2 \|_2$. We next bound them in~\lemref{naiveN_1}.

\begin{lemma} \label{lem:naiveN_1}
	Let $\rho=\sqrt{(1+\alpha)\log(1/\delta)}/\varepsilon$ and $\bN_1 \sim \cN(0,\rho^2)^{m \times t}$. Then with probability $1-2\delta$ over $\bPhi \sim \cD_R$, $ \| \bN_1([\bV]_k^\sT \bPhi)^\dagger [\bV]_k^\sT \|_2= O( \rho (1- \alpha)^{-1/2}( \sqrt{k} +\sqrt{m}))$ and $\| \bN_2 \bS \|_2 = O( \rho ( \sqrt{v} +\sqrt{n})).$
 \end{lemma}
 \begin{proof}
We first bound $ \| \bN_1([\bV]_k^\sT \bPhi)^\dagger [\bV]_k^\sT \|_2$. Let $\bC=\bN_1([\bV]_k^\sT \bPhi)^\dagger [\bV]_k^\sT$. Then $\bC \bPhi = \bN_1([\bV]_k^\sT \bPhi)^\dagger [\bV]_k^\sT \bPhi$. Now $([\bV]_k \bPhi)^\dagger [\bV]_k \bPhi$ is a projection unto a random subspace of dimension $k$.  Since every entries of $\bN_1$ is picked i.i.d. from $\cN(0,\rho^2)$, $\bC \bPhi = \bN_1([\bV]_k^\sT \bPhi)^\dagger [\bV]_k^\sT \bPhi = \begin{pmatrix} \widetilde{\bN}_1 & \mathbf{0} \end{pmatrix}$, where $\widetilde{\bN}_1$ is an $m \times k$ matrix with every entries picked i.i.d. from $\cN(0,\rho^2)$. Using Rudelson and Vershynin~\cite[Proposition 2.4]{RV10}, we have $\| \bC \bPhi \|_2 = O(\rho(\sqrt{m} + \sqrt{k}))$ with probability $99/100$. 
For our choices of  $t$, Sarlos~\cite{Sarlos06} showed that, for any  matrix $\bD$, all the singular values of $\bD \bPhi$ lies in between $(1 \pm \alpha)$ of the singular values of $\bD$. 
This in particular implies that  $\| \bC\bPhi  \|_2 \geq (1 - \alpha)^{1/2} \| \bC \|_2$, and, therefore, $\| \bC \|_2 \leq (1 - \alpha)^{-1/2} O(\rho ( \sqrt{k} +\sqrt{m})).$ This completes the proof of~\lemref{naiveN_1}.
 For the second part, if we instantiate $\bS$ with a subsampled Hadamard matrices, it is known that it satisfies $(\alpha^2,\delta)$-{\sf JLP} for the values of $v$. Applying~\lemref{SRHTinverse}, we have $\|\bS^\dagger \bN_2  \|_2 = \| \bN_2 \|_2 = O(\rho (\sqrt{n} + \sqrt{v}))$ using Rudelson and Vershynin~\cite{RV10} with probability $99/100$. This completes the proof of~\lemref{naiveN_1}.
\end{proof}

Combining~\lemref{complete} and~\lemref{naiveN_1} completes the correctness proof of~\thmref{naive}. We next prove the privacy part of~\thmref{naive}.
\begin{lemma} \label{lem:private}
 If $\rho ={\sqrt{(1+\alpha)\ln(1/\delta)}}/{\varepsilon}$, then  publishing $\bY_r$ and $\bZ$ preserves$(2\varepsilon,4\delta)$-differential privacy. \end{lemma}
\begin{proof}
First note that Clarkson and Woodruff~\cite{CW13} showed that for  the choice of $t$ and $v$,  with probability $1-\delta$, we have $\| \bS \mathbf{D} \|_F^2 \leq (1 \pm \alpha) \|\mathbf{D}\|_F^2$ and $\|  \mathbf{D}\bPhi  \|_F^2 \leq (1 \pm \alpha) \|\mathbf{D}\|_F^2$ for all $\mathbf{D}$. Let $\bA$ and $\bA'$ be two neighboring matrices such that $\mathbf{E}= \bA - \bA' $. Then $\| \bS  \mathbf{E}  \|_F^2 \leq (1+\alpha) \| \mathbf{E} \|_F^2 \leq (1+\alpha)$. Publishing $\bZ$ preserves $(\varepsilon,\delta)$-differential privacy follows from considering the vector form of the matrix $\bS {\bA} $ and $\bN_2$ and applying~\thmref{gaussian}. Similarly, we use~\thmref{gaussian} and the fact that,  for any matrix $\bC$ of appropriate dimension, $\|  \bC \bPhi \|^2_F \leq (1 + \alpha) \| \bC \|_F^2$, to prove that publishing $ {\bA}\bPhi  + \bN_1$ preserves differential privacy. 
\end{proof}

Combining~\lemref{complete},~\lemref{naiveN_1}, and~\lemref{private} gives~\thmref{naive}.
\end{proof}

In the most natural scenarios, $k \ll \max \set{m,n}$ and $\alpha$ is a small constant. Then we have $(1-\alpha)^{-1} \approx (1+\alpha)$. If we scale the value of $\alpha$ appropriately, then we have the following corollary.
\begin{corollary} \label{cor:spectral}
	Under the assumptions of~\thmref{naive} and $ \alpha \in (0,1)$  a small constant, given an $m \times n$ matrix $\bA$ in the turnstile update model,  the algorithm {\scshape Spectral-LRF}, presented in~\figref{private}, is $(\varepsilon,\delta)$-differentially private under $\priv_1$ and outputs  a $k$-rank factorization $\bU_k,$ $\bSigma_k,$ and $\bV_k^{\mathsf T}$ (with $\mathbf{M}_k=\bU_k \bSigma_k \bV_k^{\mathsf T} $), such that, with probability $9/10$ over the coin tosses of {\scshape Spectral-LRF},
	 \begin{align*} \| \bA - \mathbf{M}_k \|_2 &\leq {(1 +\alpha)} \Delta_k(\bA)
	 + {O}\paren{  \naivecor \sqrt{\log (1/\delta)}  },\end{align*}
\end{corollary}

As another corollary of~\thmref{naive}, when $\varepsilon \rightarrow \infty$, we get a non-private algorithm that computes low-rank factorization under turnstile update model with respect to the spectral norm.
\begin{corollary}  \label{cor:nonprivate}
Let $m,n \in \N$ and $k$ be the input parameters. Let $\alpha \in (0,1)$ be a constant. Given a matrix $\bA \in \R^{m \times n}$ in the turnstile update model, there exists an  algorithm that  outputs a $k$-rank factorization $\bU_k,\bSigma_k,\bV_k$ such that with probability $99/100$ over the coin tosses of the algorithm, $\| \bA - \bU_k \bSigma_k \bV_k^\sT \|_2 \leq  {(1 +\alpha )}\Delta_k(\bA)$. 
\end{corollary}

\section{Improving the Space Bound  in the Turnstile Model} \label{sec:space}
In this section, we improve the space bound of {\scshape Spectral-LRF}  under $\mathsf{Priv}_2$ in the turnstile update model. This privacy level is also stronger than that considered by previous $(\varepsilon,\delta)$-differentially private algorithms, but weaker than that in~\secref{naive}. Our algorithm formalizes the ideas mentioned in~\secref{contributions}. In particular, we show that we can simultaneously compute two sets of $k$-orthonormal basis vectors that  approximates  the top-$k$ left and top-$k$ right singular vectors of the private matrix in the turnstile update model.

We have to be careful with how we introduce noise to preserve privacy. If we use output perturbation to compute the sketches $\bY_c= {\bA} {\bPhi} + \bN$ and $\bY_r = \bPsi {\bA} + \bN'$, then one of the error terms is $\| \bN ([\bA]_k^{\mathsf T} \bPhi)^{\dagger} (\bPsi \bA \bPhi ([\bA]_k^{\mathsf T} \bPhi)^{\dagger})^\dagger \mathbf{N}' \|_2$. This term can cause the additive error to be $\Omega(n)$  if the top singular values of $\bA$ is $1/n$. If we only use input perturbation of the sketches  followed by a multiplication by Gaussian matrices as  in~\cite{BBDS12, Upadhyay14,Upadhyay14b}, the multivariate Gaussian distribution corresponding to one of the sketches is not defined  as one of the sketches is not a full rank matrix. Moreover, we cannot guarantee that the kernel space are the same for neighboring matrices $\bA$ and $\bA'$. Therefore, the privacy proof would not follow.  Our algorithm uses both the input perturbation with a careful choice of parameters and output perturbation to the other two sketches. This preserves privacy as well as keep the additive error bounded. We show the following theorem.

\begin{figure} [t]
\begin{center} 
\fbox
{
\begin{minipage}[l]{5.5in}
\small{
\medskip \noindent \textbf{Initiialization.} 
	Let $\eta:=\max\set{k^2,1/\alpha}$, $\rho_1:={\sqrt{(1+\alpha) \ln(1/\delta)}}/{\varepsilon}$,  $t=O(\eta \alpha^{-1} \log(k/\alpha) \log (1/\delta))$, $v=O(\eta \alpha^{-3} \log (t/\alpha)  \log(1/\delta))$, $\kappa=(1+\alpha)/(1-\alpha)$, and $\sigma_{\mathsf min}=16 \ln(1/\delta) \sqrt{t \kappa \ln(4/\delta)}/\varepsilon$, and $\rho_2:=\sqrt{(1+\alpha)}\rho_1$.	Set $\bA = \begin{pmatrix} \mathbf{0} & \sigma_{\mathsf min} \I_m \end{pmatrix}$.
	Sample $\bN_1 \sim \cN(0,\rho_1^2)^{t \times (m+n)}$ and $\bN_2 \sim \cN(0,\rho_2^2)^{v \times v}$. 
	Let $\bPhi \in \R^{(m+n) \times m}, \bPsi \in \R^{t \times m}$ be such that $\bPhi, \bPsi \sim \cD_R$ satisfies the statement of~\lemref{mult} for the error parameter $\sqrt{\alpha/k}$.
	Let  $\bS \in \R^{v \times m}, \bT \in \R^{(m+n) \times v}$ such that $\bS, \bT$ satisfies $({\alpha^2},\delta)$-{\sf JLP}. 
	Sample $\bOmega \sim \cN(0,1)^{m \times t}$ and $\h{\bPhi}=\frac{1}{\sqrt{t}} \bPhi \bOmega$.  Initialize $\h{\bY_c} = \bA \h{\bPhi} \in \R^{m \times t}, \h{\bY}_r= \bPsi \bA \in \R^{t \times (m+n)}$, and $\h{\bZ} = \bS \bA \bT \in \R^{v \times v}$.

\medskip
\medskip \noindent \textbf{Update rule and computing the factorization.} Suppose at time $\tau$, the stream is $(i_\tau,j_\tau,s_{\tau})$, where $(i_\tau,j_\tau) \in [m] \times [n]$. Let $\bA_\tau$ be an $m \times (m+n)$ matrix with the only non-zero entry $s_\tau$ in the position $(i_\tau,j_\tau)$. Update the matrices by the following rule: $\h{\bY}_c \leftarrow \h{\bY}_c + \bA_\tau \h{\bPhi}, \h{\bY}_r \leftarrow \h{\bY}_r +  \bPsi \bA_\tau $ and  $\h{\bZ} \leftarrow \h{\bZ} +  \bS\bA_\tau \bT$. 

\medskip
Once the matrix is streamed, we follow the following steps.
\begin{enumerate}
	\item Compute a matrix $\bU \in \R^{m \times t}$ whose columns are an orthonormal basis  for the column space of $\h{\bY}_c$.
	\item Compute a matrix $\bV \in \R^{(m+n) \times t}$ whose rows are an orthonormal basis  for the row space of $\h{\bY}_r + \bN_1$.
	\item Let $\widetilde{\bU}_s \widetilde{\bSigma}_s \widetilde{\bV}_s^{\mathsf T}$ be the SVD of $\mathbf{S} \bU \in \R^{v \times t}$. Let  $\widetilde{\bU}_t \widetilde{\bSigma}_t \widetilde{\bV}_t^{\mathsf T}$ be the SVD of $\mathbf{T}^{\sT} \bV \in \R^{v \times t}$.  
	\item Let  $\bU' \bSigma' \bV'^{\mathsf T}$ be the SVD of $ \widetilde{\bV}_s \widetilde{\bSigma}_s^{\dagger} \widetilde{\bU}_s^{\mathsf T} [\widetilde{\bU}_s \widetilde{\bU}_s^{\mathsf T}(\h{\bZ}+\bN_2) \widetilde{\bV}_t \widetilde{\bV}_t^{\mathsf T} ]_k  \widetilde{\bV}_t \widetilde{\bSigma}_t^{\dagger} \widetilde{\bU}_t^{\mathsf T} $. 
	\item Output $\bU_k=\bU \bU'$, $\bSigma_k:=\bSigma'$ and $\bV_k:=\bV \bV'$. Let $\mathbf{M}_k = \bU_k \bSigma_k \bV_k^{\mathsf T}.$
\end{enumerate}
}\end{minipage}
} \caption{Differentially private Low-rank Factorization With Improved Space ({\scshape Low-Space-LRF})} \label{fig:space}
\end{center}
\end{figure}

\begin{theorem} \label{thm:space}
 \spacetheorem
 \end{theorem}

  For the sake of simplicity, we present the algorithm when $m \leq n$. When $m >n$, we change the algorithm in~\figref{space} as follows: (i) use $\bA^{\mathsf T}$ instead of $\bA$, (ii) use  $\begin{pmatrix} \mathbf{0} & \sigma \I_n \end{pmatrix}^\sT$ instead of $\begin{pmatrix} \mathbf{0} & \sigma \I_m \end{pmatrix}$ to initialize $\h{\bY}_c,\h{\bY}_r$ and $\h{\bZ}$, (iii) the matrices $\bS, \bT, \bPhi,$ and $\bPsi$ are picked using the same distribution but with appropriate dimensions to allow matrix multiplication. 

The space complexity of the algorithm is straightforward from the choice of $t$ and $v$. 
We first prove part~\ref{spacecorrectness} of~\thmref{space}. Our proof consists of few main lemmas (\lemref{first}, \lemref{space}, \lemref{third}, and \lemref{spaceN_1}).
~\lemref{first} bounds  $\| \mathbf{M}_k - \begin{pmatrix} \bA & \mathbf{0} \end{pmatrix} \|_2$ by $ \| \mathbf{M}_k - \widehat{\bA} \|_2$ and  a fixed additive term.
~\lemref{space}  bounds $ \| \mathbf{M}_k - \widehat{\bA} \|_2$ by $\| \h{\bA} - [\h{\bA}]_k \|_2$ and two additive terms.
~\lemref{third}  bounds $\| \h{\bA} - [\h{\bA}]_k \|_2$ in the terms of $\| {\bA} - [{\bA}]_k \|_2$ and a fixed additive term. 
~\lemref{spaceN_1}  bounds the two additive terms in~\lemref{space} completing the proof. 

For all our correctness proofs, we make a standard assumption that $\delta \ll 1/100$.  Let $\bZ:=\h{\bZ}+\bN_2$, $\bY_c:=\h{\bY}_c$, $\bY_r:=\h{\bY}_r+\bN_1$, and $\bB={\h{\bA}}+\bS^\dagger  \bN_2 \bT^\dagger$, where all the variables on the right hand side of the expressions are as in~\figref{space}.

We start by proving a bound on   $\| \mathbf{M}_k - \bA \|_2$ by $ \| \mathbf{M}_k - \widehat{\bA} \|_2$ and  a small additive term. The following lemma provides such a bound.
\begin{lemma} \label{lem:first}
 Let $\bA$ be an $m \times n$ input matrix, and let $\widehat{\bA}=  \begin{pmatrix} \bA & \sigma_\mathsf{min} \I_m \end{pmatrix}$ for $\sigma_\mathsf{min}$ defined in~\figref{space}. Let $\mathbf{M}_k= \bU_k \bSigma_k \bV_k^{\mathsf T}$, where $\bU_k, \bSigma_k, \bV_k$ is the factorization outputted by  {\scshape Low-Space-LRF}. Then 
\begin{align*} \| \mathbf{M}_k - \begin{pmatrix} \bA & \mathbf{0} \end{pmatrix}\|_2  \leq  \| \mathbf{M}_k - \widehat{\bA} \|_2 + \sigma_{\min}.\end{align*}
\end{lemma}
\begin{proof}
Since $\widehat{\bA}=  \begin{pmatrix} \bA & \sigma_\mathsf{min} \I_m \end{pmatrix}$, the lemma is immediate from the following.
\begin{align*} 
\| \mathbf{M}_k - \begin{pmatrix} \bA & \mathbf{0} \end{pmatrix} \|_2 - \sigma_{\min} \| \I_m  \|_2 
\leq  \| \mathbf{M}_k - \begin{pmatrix} \bA & \mathbf{0} \end{pmatrix} - \begin{pmatrix}  \mathbf{0} & \sigma_{\min} \I \end{pmatrix} \|_2 
=  \| \mathbf{M}_k - \widehat{\bA} \|_2, \end{align*}
where the first inequality follows from the sub-additivity of the spectral norm.
\end{proof}

We next prove a bound on $\| \mathbf{M}_k - \h{\bA} \|_2 $ by  $ \| \h{\bA} - [\h{\bA}]_k\|_2$ and some additive terms. 

\begin{lemma} \label{lem:space}
Let $\mathbf{M}_k := \bU_k \bSigma_k \bV_k^{\mathsf T}$, where $\bU_k, \bSigma_k, \bV_k$ is the factorization outputted by  {\scshape Low-Space-LRF} presented in~\figref{space}. Then with probability $1-O(\delta)$ over $\bPhi ,\bPsi\sim \cD_R$ and $\bS,\bT \sim \cD_A$,
\begin{align*}& \| \mathbf{M}_k - \h{\bA} \|_2  \leq \frac{(1+\alpha)^2}{(1-\alpha)^4}  \| \h{\bA} - [\h{\bA}]_k\|_2 
+ \frac{2 (1+\alpha^2) \| \bN_2\bS^\dagger  \|_2}{(1-\alpha)^2}  + \paren{\frac{1+\alpha}{1-\alpha}}^2  \| \bN_1([\bV]_k^\sT \bPhi)^\dagger [\bV]_k^\sT \|_2.   \end{align*}
\end{lemma}
 
\begin{proof}
Our proof consists of two main claims (\claimref{firstspace} and~\claimref{thirdspace}). ~\claimref{firstspace} shows that $\bU$ and $\bV$ are ``faithful" representations of the top-$k$ left and right singular vectors of $[\bA]_k$ and~\claimref{thirdspace} shows that $\bS$ and $\bT$ satisfies the required properties with the choice of error parameters.
\begin{claim} \label{claim:firstspace}
	Let $\bA$ be the input matrix. Let $\bPhi \sim \cD_R$ be a random matrix that satisfies~\lemref{mult} with error parameters $(\sqrt{\alpha}/k,\delta)$ with respect to $[\bV]_k$ and $\bPsi \sim \cD_R$ satisfies~\lemref{mult} with error parameters $(\sqrt{\alpha}/k,\delta)$ with respect to $\h{\bA} \bPhi( [\bV]_k^\sT  \bPhi )^\dagger$.  
	Then with probability $1-2\delta$ over the choices of $\bPhi$ and $\bPsi$,
	\begin{align*} 	
	 \min_{\bX, \mathsf{r}(\bX) \leq k} \| \bU \bX \bV^\sT  - \bB \|_2& \leq (1-\alpha)^{-2} \| \h{\bA} - [\h{\bA}]_k \|_2  
	 + \| \h{\bA} \bPhi   \bD   \bN_1 \|_2  + \| \bS^\dagger \bN_2  \bT^\dagger \|_2.
 \end{align*}
\end{claim}
\begin{proof}
Let us denote by $[\bU]_k [\bSigma]_k [\bV]_k^{\mathsf T}$  the SVD of $[\h{\bA}]_k$. By construction, $\bU$ is an orthonormal basis for the column space of $\bY_c$ and $\bV$ is an orthonormal basis for the row space of $\bY_r$, i.e., 
	\begin{align} 
		\min_{\bX, \mathsf{r}(\bX) \leq k} \| \bU \bX \bV^\sT  - \bB \|_2 \leq  \min_{\bX, \mathsf{r}(\bX) \leq k} \| \bY_c \bX \bY_r - \bB \|_2. \label{eq:spacebasis2}  
	\end{align}
	
	We also note that the choice of $t$ allows $\bPhi$ and $\bPsi$ to satisfy~\lemref{mult} with error parameters $(\sqrt{\alpha}/k,\delta)$~\cite{CW13}. 
Therefore, in order to show that the matrix $\bU$ approximates the top-$k$ left singular vectors of the matrix $\h{\bA}$ and $\bV$ approximates the top-$k$ right singular vectors of the matrix $\h{\bA}$, we need to find a matrix $\bX$ such that $\| \bY_c \bX \bY_r - \h{\bA} \|_2$ is bounded by $(1- \alpha)^{-2} \|\h{\bA} - [\h{\bA}]_k \|_2$ up to some additive terms. Towards this goal, first consider the following two optimization problems:
\begin{align} \min_{\bX} \| \bPhi^\sT ([\bV]_k \bX - \h{\bA}^\sT) \|_2 \quad \text{and} \quad \min_{\bX} \| ([\bV]_k \bX - \h{\bA}^\sT) \|_2. \label{eq:regressionproblem2}
\end{align}

with
	\begin{align*} 
	& \t{\bX}:=\argmin_{\bX} \| \bPhi^\sT([\bV]_k \bX - \h{\bA}^\sT) \|_2 \\
	& \h{\bX}:=\argmin_{\bX} \| ([\bV]_k \bX - \h{\bA}^\sT) \|_2  \end{align*}
being one of the solutions of~\eqnref{regressionproblem2}.

Recall that $\bPhi$ satisfies~\lemref{mult} with respect to $[\bV]_k$. Therefore, if we substitute $\bA$ by $\h{\bA}$ in~\eqnref{aim}, then $\| [\bV]_k \t{\bX} - \h{\bA}^\sT \|_2 \leq  (1-\alpha)^{-1} \| ([\bV]_k \h{\bX} - \h{\bA}^\sT) \|_2 $, where $\t{\bX} = (\bPhi^\sT [\bV]_k)^\dagger \bPhi^\sT \h{\bA}^\sT.$ This is the same as 
\begin{align*}
\|  \h{\bA} \bPhi ([\bV]_k^\sT \bPhi)^\dagger & [\bV]_k^\sT - \h{\bA} \|_2 =  \| \t{\bX}^\sT [\bV]_k^\sT  - \h{\bA} \|_2 
\leq  (1-\alpha)^{-1} \| ([\bV]_k \h{\bX} - \h{\bA}^\sT) \|_2 \end{align*}

Define a matrix $\mathbf{C}: =\h{\bA} \bPhi( [\bV]_k^\sT  \bPhi )^\dagger .$ The matrix $\bC$ has a rank at most $k$ and the above equation can be rewritten as
\begin{align}
	\| \bC [\bV]_k^{\mathsf T} - \h{\bA} \|_2 \leq (1-\alpha)^{-1}  \| ([\bV]_k \h{\bX} - \h{\bA}^\sT) \|_2. \label{eq:relate}
\end{align}
 
 Now  consider the following optimization problem:
\begin{align} \min_{\bX} \| \bPsi  (\bC \bX - \h{\bA}) \|_2  \quad \text{and} \quad \min_{\bX} \| (\bC \bX - \h{\bA}) \|_2. \label{eq:regressionproblem1}
\end{align}

Setting $\mathbf{L}:=\bPhi \bC$ and $\bA = \h{\bA}$~\thmref{FT07} and a singular value decomposition of $ \bPsi \bC$, we have
	\begin{align*} \t{\bX}':= (\bPsi \bC)^\dagger \bPsi \h{\bA} =\argmin_{\bX} \| \bPsi (\bC \bX - \h{\bA}^\sT)   \|_2 . \end{align*}

Recall that $\bPsi$ satisfies~\lemref{mult} with respect to $\bC$. Therefore,  using~\eqnref{aim} on problems~(\ref{eq:regressionproblem1}), we have
\begin{align}
	\| \bC \t{\bX}' - \h{\bA} \|_2 
			&\leq   (1-\alpha)^{-1}\min_{\bX} \| \bC {\bX} - \h{\bA} \|_2 \nonumber  \\
			& \leq  (1-\alpha)^{-1} \|  \bC [\bV]_k^\sT - \h{\bA} \|_2 \nonumber  \\
			& \leq (1-\alpha)^{-2} \| ([\bV]_k \h{\bX} - \h{\bA}^\sT) \|_2 \nonumber \\
			&=(1-\alpha)^{-2} \min_\bX \| ([\bV]_k {\bX} - \h{\bA}^\sT) \|_2  \nonumber  \\
			&\leq (1-\alpha)^{-2} \| [\bU]_k  [\bSigma]_k [\bV]_k^{\mathsf T} - \h{\bA} \|_2  \nonumber  \\
			& = (1-\alpha)^{-2} \| [\h{\bA}]_k - \h{\bA} \|_2. \label{eq:spacetranspose2}  
\end{align}
In the above, the third inequality follows from using~\eqnref{relate} and fourth follows from the definition of $\h{\bX}$.

	Let $\bD=  ([\bV]_k^\sT \bPhi)^\dagger (\bPsi \h{\bA} \bPhi ([\bV]_k \bPhi)^\dagger )^\dagger $, i.e., $\bC \t{\bX}' = \h{\bA} \bPhi \bD   \bPsi \h {\bA} $. Since $([\bV]_k \bPhi)^{\dagger} [\bV]_k $ has rank at most $k$ and $\bY= \bPsi \h{\bA} + \bN_1$, with probability $1-\delta$ over $\bPhi \sim \cD_R$,  
	\begin{align} 
		 \min_{\bX, \mathsf{r}(\bX) \leq k} \| \bY_c  \bX \bY_r  - \bB \|_2 
		&\leq  \| \h{\bA} \bPhi  \bD  \bPsi \h {\bA} +  \h{\bA} \bPhi   \bD   \bN_1 - \bB \|_2 \nonumber \\
		&\quad =  \| \h{\bA} \bPhi  \bD   \bPsi \h {\bA}  +  \h{\bA} \bPhi   \bD   \bN_1 - \h{\bA} - \bS^\dagger \bN_2    \bT^\dagger \|_2 \nonumber \\
		&\quad \leq  \| \h{\bA} \bPhi \bD   \bPsi \h {\bA} - \h{\bA} \|_2 + \| \h{\bA} \bPhi   \bD   \bN_1 \|_2 + \| \bS^\dagger \bN_2   \bT^\dagger  \|_2 \nonumber \\
		&\quad \leq  \| \bC \t{\bX}' - \h{\bA} \|_2 + \| \h{\bA} \bPhi   \bD   \bN_1 \|_2 + \| \bS^\dagger \bN_2   \bT^\dagger  \|_2\label{eq:spaceset2}  
	\end{align}
	
	Combining~\eqnref{spacetranspose2} and~\eqnref{spaceset2}, we have with probability $1-2\delta$ over the choice of $\bPhi $ and $\bPsi$, 
	\begin{align}  \min_{\bX, \mathsf{r}(\bX) \leq k} & \| \bY_c  \bX \bY_r  - \bB \|_2 \leq (1-\alpha)^{-2}  \| \h{\bA} - [\h{\bA}]_k \|_2 
	+ \| \h{\bA} \bPhi   \bD   \bN_1 \|_2  + \| \bS^\dagger \bN_2     \bT^\dagger\|_2. \label{eq:spacestep12} \end{align}
	
	Combining~\eqnref{spacestep12} and~\eqnref{spacebasis2}, we have with probability $1-3\delta$ over $\bPhi $ and $\bPsi$, 
	\begin{align} 
	\min_{\bX, \mathsf{r}(\bX) \leq k} & \| \bU \bX \bV^\sT  - \bB \|_2 \leq (1-\alpha)^{-2} \| \h{\bA} - [\h{\bA}]_k \|_2 
	+ \| \h{\bA} \bPhi   \bD   \bN_1 \|_2  + \| \bS^\dagger \bN_2  \bT^\dagger \|_2. \label{eq:step22space}
	\end{align}
This completes the proof of Claim~\ref{claim:firstspace}.
\end{proof}	

We would have been done if we had not generated a noisy version of the sketch $\bS \h{\bA} \bT$. The next claim bounds the effect of generating this sketch. 

\begin{claim} \label{claim:thirdspace}
Let $\bU, \bB, \h{\bA}, \bS, \bN_1$, and $\bN_2$ be as  above, and let $\widetilde{\bX}_k= \argmin_{\bX, \mathsf{r}(\bX)=k} \| \bS(\bU \bX \bV^\sT -\bB) \bT \|_2$. Let $\cD_A$ be a distribution that satisfies  $(\alpha^2,\delta)$-{\sf JLP}.  Then  with probability $1-6\delta$ over $\bS, \bT \sim \cD_A$ and $\bPhi, \bPsi \sim \cD_R$, we have
\begin{align*}	\paren{\frac{1-\alpha}{1+ \alpha}}^2 & {\|( \bU \widetilde{\bX}_k \bV^\sT  - \bB) \|_2} \leq (1-\alpha)^{-2}  \| \h{\bA} - [\h{\bA}]_k \|_2 
+ \|   \bS^\dagger \bN_2\bT^\dagger  \|_2  + \| \h{\bA} \bPhi   \bD   \bN_1 \|_2 . 
\end{align*}
\end{claim}
\begin{proof}
	Note that the choice of $v$ allows $\bS$ and $\bT$ to satisfy $({\alpha}^2,\delta)$-{\sf JLP}. 

Let $\h{\bX}_k =  \argmin_{\bX, \mathsf{r}(\bX) \leq k} \| (\bU \bX  \bV^\sT- \bB) \|_2.$ We want to bound $\| \bU \h{\bX}_k  \bV^\sT - \bB \|_2$ in terms of $\| \bU \t{\bX}_k  \bV^\sT- \bB\|_2$. 
{Using the definitions of $\t{\bX}_k$ and $\h{\bX}_k$ along with~\eqnref{S}, we have	}
\begin{align}	\min_{\bX, \mathsf{r}(\bX)=k}  \| \bU \bX \bV^\sT  -\bB \|_2 &= \| \bU \widehat{\bX}_k \bV^\sT  - \bB \|_2 \nonumber \\
			& \geq (2 + 2\alpha)^{-1} \| \bS(\bU  \widehat{\bX}_k \bV^\sT - \bB )  \|_2\nonumber \\
			& \geq (2 + 2\alpha)^{-2} \| \bS(\bU  \widehat{\bX}_k \bV^\sT - \bB ) \bT \|_2\nonumber \\
			& \geq (2 + 2\alpha)^{-2}  \min_{\bX, \mathsf{r}(\bX)\leq k} \| \bS(\bU {\bX}\bV^\sT - \bB ) \bT \|_2 \nonumber \\
			& = (2 + 2\alpha)^{-1}  \| \bS(\bU \widetilde{\bX}_k \bV^\sT - \bB  ) \bT \|_2 \nonumber\\
			& \geq  \paren{\frac{1-\alpha}{1+\alpha}}^2  \| (\bU \widetilde{\bX}_k \bV^\sT - \bB  ) \|_2 
			.  \label{eq:S2space}  
	\end{align}
	
	Combining~\eqnref{S2space} with~\eqnref{step22space},  we have with probability $1-3\delta$ over $\bPhi \sim \cD_R$ and $\bS \sim \cD_A$, 
	\begin{align}
	\paren{\frac{1-\alpha}{1+\alpha}}^2 & \|( \bU \widetilde{\bX}_k \bV^\sT- \bB) \|_2 \leq (1-\alpha)^{-2}  \| \h{\bA} - [\h{\bA}]_k \|_2 
		+  \| \bS^\dagger \bN_2   \bT^\dagger  \|_2  + \| \h{\bA} \bPhi   \bD   \bN_1 \|_2 . \label{eq:step32} 
	\end{align}
This completes the proof of Claim~\ref{claim:thirdspace}.
\end{proof}


We can now complete the proof of~\lemref{space}. Since $\bB= \h{\bA} + \bS^\dagger \bN_2  \bT^\dagger$
 \begin{align*}
 	   \| \bU \widetilde{\bX}_k \bV^\sT  - \h{\bA} \|_2  - \| \bS^\dagger \bN_2 \bT^\dagger  \|_2
	  		& \leq  \| \bU \widetilde{\bX}_k\bV^\sT  - \bB \|_2 
			\\ &\leq \frac{(1+\alpha)^2}{(1-\alpha)^4}  \|\h{\bA} - [\h{\bA}]_k\|_2 
			+  \paren{\frac{1+\alpha}{1-\alpha}}^2  (\| \bS^\dagger \bN_2  \bT^\dagger  \|_2  +\| \h{\bA} \bPhi   \bD   \bN_1 \|_2 ).
  \end{align*}
{This implies that} 
 \begin{align*}
& \| \bU \widetilde{\bX}_k \bV^\sT    - \h{\bA} \|_2 
 			\leq \frac{(1+\alpha)^2}{(1-\alpha)^4}  \|\h{\bA} - [\h{\bA}]_k\|_2 
			+ \frac{2 (1+\alpha^2) \| \bN_2\bS^\dagger \bT^\dagger \|_2}{(1-\alpha)^2}  + \paren{\frac{1+\alpha}{1-\alpha}}^2 \| \h{\bA} \bPhi   \bD   \bN_1 \|_2 .  
  \end{align*}

Setting $\mathbf{O}:=\bZ$, $\mathbf{L}:=\bS \bU$ and $\mathbf{R}:=\bV^\sT \bT$ in the statement of~\thmref{FT07}, we have \begin{align*}\t{\bX}_k= \widetilde{\bV}_s \widetilde{\bSigma}_s^{\dagger} \widetilde{\bU}_s^{\mathsf T} [\widetilde{\bU}_s \widetilde{\bU}_s^{\mathsf T}\bZ \widetilde{\bV}_t \widetilde{\bV}_t^{\mathsf T} ]_k  \widetilde{\bV}_t \widetilde{\bSigma}_t^{\dagger} \widetilde{\bU}_t^{\mathsf T}\end{align*} 
is one of the candidate solutions to  $$\argmin_{\bX, \mathsf{r}(\bX)=k} \| \bS(\bU \bX \bV^\sT -\bB) \|_2.$$
 Using this value of $\t{\bX}_k$ completes the proof of~\lemref{space}. 
\end{proof}
 
 \lemref{space} contains two additive terms $\| \bS^\dagger \bN_2  \bT^\dagger  \|_2$ and $\| \h{\bA} \bPhi   \bD   \bN_1 \|_2$. We next bound both these terms. 


\begin{lemma} \label{lem:spaceN_1}
	Let $\rho=\sqrt{(1+\alpha)\log(1/\delta)}/\varepsilon$ and $\bN_1 \sim \cN(0,\rho^2)^{(m+n) \times t}$. Then with probability $1-2\delta$ over $\bPhi \sim \cD_R$, $ \| \h{\bA} \bPhi   \bD   \bN_1 \|_2= O( \rho_1 (1-\alpha)^{-1/2}( \sqrt{k} +\sqrt{m+n}))$ and $\|\bS^\dagger \bN_2 \bT^\dagger \|_2 =  O(\rho_2 \sqrt{v}).$
 \end{lemma}
 \begin{proof}
Let $\bC= \h{\bA} \bPhi   \bD   \bN_1$. Then $$\bPsi \bC  = \bPsi \h{\bA} \bPhi  ([\bV]_k^\sT \bPhi)^\dagger (\bPsi \h{\bA} \bPhi ([\bV]_k \bPhi)^\dagger )^\dagger  \bN_1 .$$ Now $\bPsi \h{\bA} \bPhi   \bD$ is a projection unto a random subspace of dimension at most $k$.  Since every entries of $\bN_1$ is picked i.i.d. from $\cN(0,\rho^2)$, $\bPsi\bC  =  \widetilde{\bN}_1$, where $\widetilde{\bN}_1$ is an $m \times k$ matrix with every entries picked i.i.d. from $\cN(0,\rho^2)$. Using Rudelson and Vershynin~\cite[Proposition 2.4]{RV10}, we have $\| \bPsi \bC  \|_2 = O(\rho(\sqrt{m+n} + \sqrt{k}))$ with probability $99/100$. 
For our choices of  $t$, Sarlos~\cite{Sarlos06} showed that, for any  matrix $\bE$, all the singular values of $\bPsi \bE $ lies in between $(1 \pm \alpha)$ of the singular values of $\bE$. 
This in particular implies that $\| \bPsi \bC  \|_2 \geq (1 - \alpha)^{1/2} \| \bC \|_2$. That is,  $\| \bC \|_2 = (1-\alpha)^{-1/2} O(\rho_1 ( \sqrt{k} +\sqrt{m+n})).$ 
For the second part, if we instantiate $\bS$ with a subsampled Hadamard matrices, it is known that it satisfies $(\alpha^2,\delta)$-{\sf JLP} for the values of $v$. Applying~\lemref{SRHTinverse} twice, we have $\|\bS^\dagger \bN_2 \bT^\dagger \|_2 = \| \bN_2 \|_2 = O(\rho_2 \sqrt{v})$ using Rudelson and Vershynin~\cite{RV10} with probability $99/100$. This completes the proof of~\lemref{spaceN_1}.
 \end{proof}

 We finally compute an upper bound on $\| \widehat{\bA} - [\widehat{\bA}]_k \|_2$ in terms of $\| {\bA} - [{\bA}]_k \|_2$. 

\begin{lemma} \label{lem:third}
Let $d$ be the maximum of the rank of $\bA$ and $\widehat{\bA}$. Let $\sigma_1,\cdots, \sigma_d$ be the singular values of $\bA$ and $\h{\sigma}_1,\cdots, \h{\sigma}_d$ be the singular values of $\widehat{\bA}$. Then $\| \h{\bA} - [\h{\bA}]_k\|_2 \leq \| {\bA} - [{\bA}]_k\|_2 + \sigma_\mathsf{min}$.
\end{lemma}
\begin{proof}
We first prove that $| \sigma_i - \sigma_i' | \leq \sigma_\mathsf{min}$ for all $1 \leq i \leq d$. We can write $\widehat{\bA} =\begin{pmatrix} \bA & \mathbf{0} \end{pmatrix} + \begin{pmatrix}  \mathbf{0} & \sigma_\mathsf{min} \I_m \end{pmatrix}.$ By construction, all the singular values of $ \begin{pmatrix}  \mathbf{0} &  \sigma_\mathsf{min} \I_m \end{pmatrix}$ are $ \sigma_\mathsf{min}$; therefore,~\thmref{Weyl} implies that $| \sigma_i - \sigma_i' | \leq \sigma_\mathsf{min}$ for all $1 \leq i \leq d$. 
In particular, since $\| \h{\bA} - [\h{\bA}]_k\|_2 = \h{\sigma}_{k+1}$ and $\| {\bA} - [{\bA}]_k\|_2 = {\sigma}_{k+1}$, this implies that $|\h{\sigma}_{k+1} - \sigma_{k+1}| \leq \sigma_\mathsf{min}$. The lemma follows.
\end{proof}

Combining~\lemref{first},~\lemref{space},~\lemref{third}, and \lemref{spaceN_1} completes the proof of part~\ref{spacecorrectness} of~\thmref{space} after noting that $\alpha \in (0,1)$. We next prove part~\ref{spaceprivacy} of~\thmref{space}. This is essentially the same as in Upadhyay~\cite{Upadhyay16}. We include it for the sake of completion.

 \begin{lemma} \label{lem:low_space_private}
 If $\sigma_\mathsf{min}, \rho_1$ and $\rho_2$ be as in~\thmref{space}, then the algorithm presented in~\figref{space}, is $(3\varepsilon,6\delta)$-differentially private.
 \end{lemma}
We prove the lemma when $m \leq n$. The case for $m \geq n$ is analogous after inverting the roles of $\widehat{\bPhi}$ and $\bPsi$. Let $\bA$ and $\bA'$ be two neighboring matrices, i.e., $\mathbf{E}= \bA - \bA' = \mathbf{u} \mathbf{v}^{\mathsf T}$. Then $\widehat{\bA}$ and $\widehat{\bA}'$, constructed by {\scshape Low-Space-}$\lrf$, has the following property: $\widehat{\bA}' = \widehat{\bA} + \begin{pmatrix}   \mathbf{E} & \mathbf{0}  \end{pmatrix}$. 
 
\begin{claim} \label{claim:Y_r}
 If $\rho_1 ={\sqrt{(1+\alpha)\ln(1/\delta)}}/{\varepsilon}$, then  publishing $\bY_r$ preserves$(\varepsilon,2\delta)$-differential privacy. If $\rho_2 ={(1+\alpha)\sqrt{\ln(1/\delta)}}/{\varepsilon}$, then  publishing  $\bZ$ preserves$(\varepsilon,2\delta)$-differential privacy. \end{claim}
\begin{proof}
First note that by the choice of $t$ and $v$, using Clarkson and Woodruff~\cite{CW13} gives that  $\| \bS \mathbf{D} \|_F^2 \leq (1 \pm \alpha) \|\mathbf{D}\|_F^2$ and $\| \bPsi \mathbf{D}  \|_F^2 \leq (1 \pm \alpha) \|\mathbf{D}\|_F^2$ for all $\mathbf{D}$ with probability $1-\delta$. Let $\bA$ and $\bA'$ be two neighboring matrices such that $\mathbf{E}= \bA - \bA' = \mathbf{u} \mathbf{v}^{\mathsf T}$. Then $\| \bS \begin{pmatrix} \mathbf{E} & \mathbf{0} \end{pmatrix} \bT^{\mathsf T} \|_F^2 \leq (1+\alpha) \|  \begin{pmatrix} \mathbf{E} & \mathbf{0} \end{pmatrix}  \bT^{\mathsf T}\|_F^2 \leq (1+\alpha)^2$. Publishing $\bZ$ preserves $(\varepsilon,\delta)$-differential privacy follows from considering the vector form of the matrix $\bS \widehat{\bA} \bT^{\mathsf T}$ and $\bN_2$ and applying~\thmref{gaussian}. Similarly, we use~\thmref{gaussian} and the fact that,  for any matrix $\bC$ of appropriate dimension, $\| \bPsi \bC  \|^2 \leq (1 + \alpha) \| \bC \|_F^2$, to prove that publishing $\bPsi \widehat{\bA}  + \bN_1$ preserves differential privacy. 
\end{proof}

We next give a proof sketch that $\bY_c$ is $(\varepsilon,\delta)$-differentially private. This would complete the proof of~\lemref{low_space_private} as the lemma would follow by combining~\lemref{post} and~\thmref{DRV10}. 
\begin{claim} \label{claim:Y_c}
Let $\sigma_{\mathsf min}$ be as in~\thmref{space}. Then  publishing $\bY_c$ preserves$(\varepsilon,2\delta)$-differential privacy. 
 \end{claim}
\begin{proof}
Our proof uses an analysis of multivariate Gaussian distribution. The multivariate Gaussian distribution is a generalization of univariate Gaussian distribution. Let $\mu$ be an $N$-dimensional vector. An $N$-dimensional multivariate random variable, $\bx \sim \cN(\mathbf{\mu}, \bLambda)$, where $\bLambda = \E [(\bx - \mu)(\bx -\mu)^{\mathsf T}]$ is the $N \times N$ covariance matrix, has the probability density function given by 
$\PDF_\mathbf{X} (\bx) := \frac{e^{- \bx^{\mathsf T}    \bLambda^\dagger   \bx/2}}{\sqrt{(2 \pi)^{\mathsf{r}(\bSigma)} \det(\bLambda)}}   \label{eq:pdfmultivariate},$ where $\det(\cdot)$ denotes the determinant of the matrix. 
 If $\bLambda$ has a non-trivial kernel space, then the multivariate distribution is undefined. However, in this proof, all our covariance matrices have only trivial kernel.
Multivariate Gaussian distributions is invariant under affine transformation, 
i.e., if $\by = \bA\bx+\mathbf{b}$, where $\bA \in \R^{M \times N}$ is a rank-$M$ matrix and $\mathbf{b} \in \R^M$, then $\by \sim \cN (\bA\mu+\mathbf{b},\bA \bLambda \bA^{\mathsf T}   )$. 

Let $\bA - \bA' = \mathbf{E}=\mathbf{u} \mathbf{v}^{\mathsf T}$ and let $\widehat{\mathbf{v}} = \begin{pmatrix} \mathbf{v} & \mathbf{0}^m \end{pmatrix}$. Then $\widehat{\bA} - \widehat{\bA}' =  \mathbf{u}\widehat{\mathbf{v}}^{\mathsf T}$. Since $\bPhi^{\mathsf T}$ is sampled from $\cD_R$, using~\cite{CW13}, we have $\| \bPhi^{\mathsf T} \mathbf{W} \|_F^2 = (1+\alpha) \| \mathbf{W} \|_F^2$ for any matrix $\bW$ with probability $1-\delta$. Therefore, $\mathbf{u} \mathbf{v}^{\mathsf T}\bPhi =  (1+\alpha)^{1/2}\mathbf{u} \widetilde{\mathbf{v}}^{\mathsf T} = \widetilde{\mathbf{u}} \widetilde{\mathbf{v}}^{\mathsf T}$ for some unit vectors $\mathbf{u}$, $\widetilde{\mathbf{v}}$ and $ \widetilde{\mathbf{u}} = (1+\alpha)^{1/2}\mathbf{u}$. 
We now show that   $ \widehat{\bA}\bPhi \bOmega_1$ preserves privacy. 
We prove that each row of the published matrix preserves $(\varepsilon_0, \delta_0)$-differential privacy for some appropriate $\varepsilon_0,\delta_0$, and then invoke~\thmref{DRV10} to prove that the published matrix preserves $(\varepsilon,\delta)$-differential privacy. 

It may seem that the privacy of $\bY_c$ follows from the result of Blocki {\it et al.}~\cite{BBDS12}, but this is not the case because of the following reasons.
\begin{enumerate}
	\item The definition of neighboring matrices considered in this paper is different from that of Blocki {\it et al.}~\cite{BBDS12}. To recall, Blocki {\it et al.}~\cite{BBDS12} considered two matrices neighboring if they differ in at most one row by a unit norm. In our case, we consider two matrices are neighboring if they have the form $\mathbf{u} \mathbf{v}^{\mathsf T}$ for unit vectors $\mathbf{u},\mathbf{v}$. 
	\item We multiply the Gaussian matrix to a random projection of $\widehat{\bA}$ and not to $\bA$ as in the case of Blocki {\it et al.}~\cite{BBDS12}, i.e., to $\widehat{\bA} \bPhi$ and not to $\widehat{\bA}$. 
\end{enumerate}

We first give a brief overview of how to deal with these issues. The first issue is resolved by analyzing $(\widehat{\bA} - \widehat{\bA}')\bPhi.$ We observe that this expression can be represented in the form of $\widetilde{\mathbf{u}} \widetilde{\mathbf{v}}^{\mathsf T}$, where $ \widetilde{\mathbf{u}} = (1+\alpha)^{1/2}\mathbf{u}$  for some  $\|\mathbf{u}\|_2=1$, $\|\widetilde{\mathbf{v}}\|_2=1$. 
The second issue can be resolved by observing that $\bPhi$ satisfies $(\alpha,\delta)$-{\scshape JLP} because of the choice of $t$. Since the rank of $\widehat{\bA}$ and $\widehat{\bA} \bPhi$ are the same, the singular values of $\widehat{\bA} \bPhi$ are within a multiplicative factor of $(1 \pm \alpha)^{1/2}$ of the singular values of $\bPhi$ with probability $1-\delta$  due to Sarlos~\cite{Sarlos06}. Therefore, our proof goes through if we scale the singular values of $\widehat{\bA}$ appropriately.
\end{proof}

In practical scenarios, $k \ll \max \set{m,n}$ and $\alpha$ is a small constant. This implies that $(1-\alpha)^{-1} \approx (1+\alpha)$. If we scale the value of $\alpha$ appropriately, then we have the following corollary:
\begin{corollary} \label{cor:space}
	Under the assumptions of \thmref{space} and $0 < \alpha <1$  a small constant, given an $m \times n$ matrix $\bA$ in the turnstile update model,  the algorithm {\scshape Spectral-LRF}, presented in~\figref{private}, is $(\varepsilon,\delta)$ differentially private  under $\priv_2$  and outputs  a $k$-rank factorization $\bU_k, \bSigma_k, \bV_k^{\mathsf T}$ (with $\mathbf{M}_k=\bU_k \bSigma_k \bV_k^{\mathsf T}$, such that, with probability $9/10$ over the coin tosses of {\scshape Low-Space-LRF},
	 \begin{align*} \|  \begin{pmatrix} \bA & \mathbf{0} \end{pmatrix} &- \mathbf{M}_k \|_2 \leq {(1 +\alpha)} \Delta_k(\bA) 
	 + O \paren{ \spacecor\sqrt{\ln(1/\delta)}} .\end{align*} 
\end{corollary}


\section{A Generic Transformation for Continual Release}
Our algorithms for the continual release use the binary tree mechanism~\cite{CSS,DNPR10} to store the sketches of matrix generated at various time epochs. We present the algorithm in~\figref{continual} for the sake of completion. 
 We give an intuition of the algorithm by visualizing  the data-structure that stores all the sketches in the form of a binary tree. Every leaf node $\tau$  stores the sketches of $\bA_\tau$, the root note stores the sketch of the entire matrix streamed in $[0,T]$, and every other node $\mathsf{n}$ stores the sketch corresponding to the updates in  a time range represented by the leaves of the subtree rooted at $\mathsf{n}$, i.e., $\h{\bY}_i$ and $\h{\bZ}_i$ stores sketches involving $2^i$ updates to $\bA$. If a query is to compute the low-rank factorization of the matrix from a particular time range $[1,\tau]$, we find the nodes that uniquely cover the time range $[1,\tau]$. We then use the value of $\bY(\tau)$ and $\bZ(\tau)$ formed using those nodes to compute the low-rank factorization. From the binary tree construction, every time epoch appears in exactly $O(\log T)$ nodes (from the leaf to the root node). Moreover, every range $[1,\tau]$ appears in at most $O(\log T)$  nodes of the tree (including leaves and root node). 
A straightforward application of the analysis of Chan {\it et al.}~\cite{CSS} to the proof of~\corref{spectral} gives us the following result.
\begin{figure} [t]
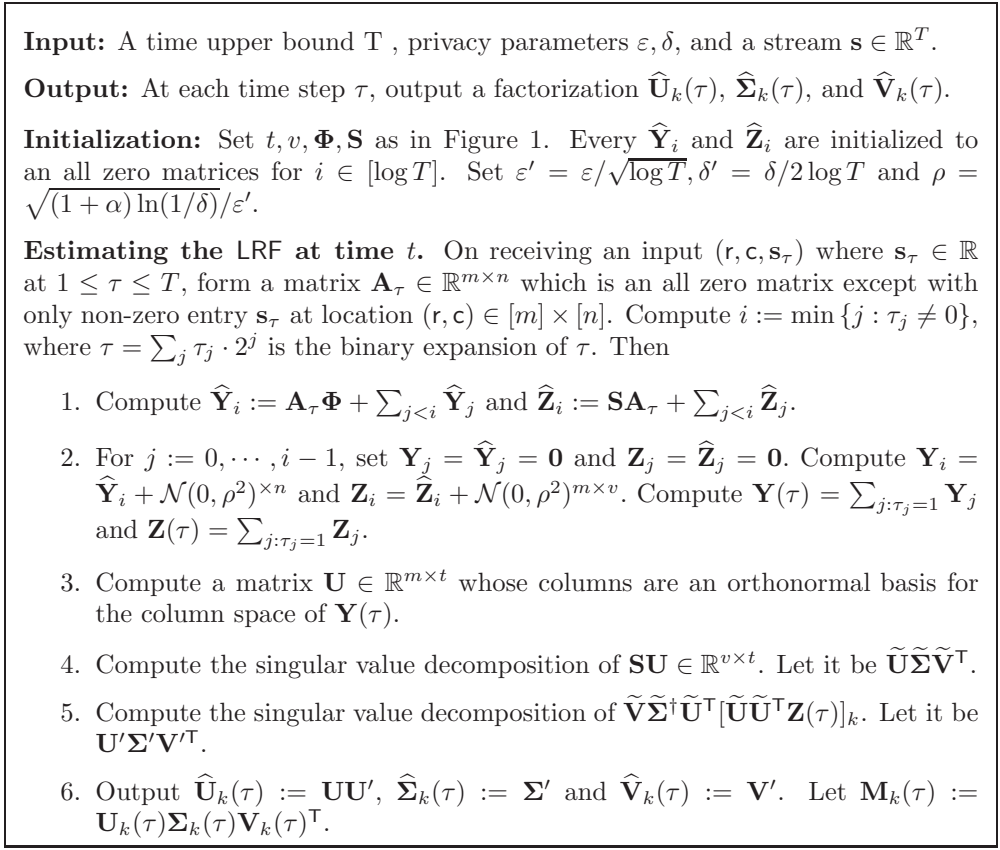

\begin{center} 
\fbox
{
\begin{minipage}[l]{5in}
\small{
\medskip \noindent \textbf{Input:} A time upper bound T , privacy parameters $\varepsilon,\delta$, and a stream $\mathbf{s} \in \R^T$. 

\medskip \noindent \textbf{Output:} At each time step $\tau$, output a factorization $\h{\bU}_k(\tau)$, $\h{\bSigma}_k(\tau)$, and $\h{\bV}_k(\tau)$.

\medskip \noindent \textbf{Initialization:} Set $t,v, \bPhi, \bS$ as in~\figref{private}.  Every $\h{\bY}_i$ and $\h{\bZ}_i$ are initialized to an all zero matrices for  $i \in [\log T]$. Set $\varepsilon' = \varepsilon/\sqrt{\log T}, \delta'=\delta/2 \log T$ and $\rho={\sqrt{(1+\alpha) \ln(1/\delta)}}/{\varepsilon'}$.

\medskip \noindent \textbf{Estimating the $\lrf$ at time $t$.}  
On receiving an input $(\mathsf{r},\mathsf{c},\mathbf{s}_\tau)$ where $\mathbf{s}_\tau  \in \R$ at $1\leq \tau \leq T$, form a matrix $\bA_\tau \in \R^{m \times n}$ which is an all zero matrix except with only non-zero entry $ \mathbf{s}_\tau$ at location $(\mathsf{r},\mathsf{c}) \in [m] \times [n]$. Compute $i := \min \set{j : \tau_j \neq 0}$, where  $\tau =\sum_j \tau_j \cdot 2^j$  is the binary expansion of $\tau$. Then 
\begin{enumerate}
	\item Compute $\h{\bY}_i := \bA_\tau \bPhi + \sum_{j <i} \h{\bY}_j$ and $\h{\bZ}_i := \bS \bA_\tau  + \sum_{j <i} \h{\bZ}_j.$ 
	\item For $j:= 0, \cdots, i-1$, set $ \bY_j = \h{\bY}_j  = \mathbf{0}$ and $\bZ_j= \h{\bZ}_j= \mathbf{0}.$ Compute $ {\bY}_i = \h{\bY}_i + \cN(0,\rho^2)^{ \times n}$ and ${\bZ}_i =\h{ \bZ}_i + \cN(0,\rho^2)^{m \times v}.$ \label{item:update}
 Compute $  {\bY}(\tau) = \sum_{j: \tau_j=1} {\bY}_j $ and $ {\bZ}(\tau) = \sum_{j: \tau_j=1} {\bZ}_j .$ \label{item:estimate}
		\item Compute a matrix $\bU \in \R^{m \times t}$ whose columns are an orthonormal basis  for the column space of $\bY(\tau)$.
	\item Compute the singular value decomposition of $\mathbf{S} \bU \in \R^{v \times t}$. Let it be $\widetilde{\bU} \widetilde{\bSigma} \widetilde{\bV}^{\mathsf T}.$ 
	\item Compute the singular value decomposition of $ \widetilde{\bV} \widetilde{\bSigma}^{\dagger} \widetilde{\bU}^{\mathsf T} [\widetilde{\bU} \widetilde{\bU}^{\mathsf T}\bZ(\tau)]_k$. 
	Let it be $\bU' \bSigma' \bV'^{\mathsf T}$. 
	\item Output $\h{\bU}_k(\tau):=\bU \bU'$, $\h{\bSigma}_k(\tau):= \bSigma'$ and $\h{\bV}_k(\tau):=\bV'$. Let $\mathbf{M}_k(\tau) := \bU_k(\tau)  \bSigma_k(\tau) \bV_k(\tau)^{\mathsf T}.$
\end{enumerate}	
}\end{minipage}
} \caption{Differentially private Low-rank Factorization Under Continual Release ({\scshape Continual-LRF})} \label{fig:continual}
\end{center}
\end{figure}

\begin{theorem} \label{thm:continual}
 	Let  $m, n \in \N$ and $\varepsilon,\delta$ be the input parameters.  Let $0 < \alpha <1$ be an arbitrary constant. Let $\bA(t) \in \R^{m \times n}$ be the matrix formed until time $t$. Then  the algorithm {\scshape Continual-LRF}, presented in~\figref{continual}, outputs  an $(\varepsilon,\delta)$-differentially private $k$-rank approximation $\mathbf{M}_k(t)$ under continual release and under $\priv_1$, such that, with probability at least $9/10$ over the coin tosses of {\scshape Continual-LRF},
	 \begin{align*} \| \bA(t) &- \mathbf{M}_k (t)\|_2 \leq \frac{(1+\alpha)}{(1-\alpha)^2} \Delta_k( \bA (t))  
	 + O \paren{ \naivecontinual \sqrt{\log (1/\delta)}  }  .\end{align*} 
\end{theorem}

We can also convert the algorithm {\scshape Low-Space-LRF} to one that outputs a low-rank factorization under continual release by using less space than {\scshape Continual-LRF} and secure under $\priv_2$. We make the following changes to {\scshape Continual-LRF}: 
(i) Initialize $(\h{\bY}_c)_i, (\h{\bY}_r)_i,$ and $(\h{\bZ})_i$ as we initialize $\bY_c, \h{\bY}_r$ and $\h{\bZ}$ in~\figref{space}  for all $i \in [\log T]$,
(ii) we maintain $ (\bY_c)_j, (\h{\bY_c})_j $, $ (\bY_r)_j, (\h{\bY_r})_j$,  $\bZ_j,$ and $\h{\bZ}_j.$ 

\begin{theorem} \label{thm:spacecontinual}
 	Let  $m, n \in \N$ and $\varepsilon,\delta$ be the input parameters.  Let $k$ be the desired rank of the factorization.  Let $0 < \alpha <1$ be an arbitrary constant. Let $\bA(t) \in \R^{m \times n}$ be the matrix formed until time $t$. Then  there is an efficient algorithm that outputs  an $(\varepsilon,\delta)$-differentially private $k$-rank approximation $\mathbf{M}_k(t)$ under continual release and under $\priv_2$, such that, with probability at least $9/10$,
	 \begin{align*} \| \bA(t) - \mathbf{M}_k (t)\|_2 & \leq \frac{(1+\alpha)}{(1-\alpha)^2} \Delta_k( \bA (t))  
	 + O \paren{ \spacecontinual {\sqrt{\log (1/\delta)} }  } .\end{align*} 
\end{theorem}

 
\bibliographystyle{plain}
{ \bibliography{spectral}}

\begin{appendix}

\section{Proof of~\claimref{Y_c}} 
We now return to the proof. Denote by $\widehat{\bA}=  \begin{pmatrix} \bA & \sigma_\mathsf{min} \I_m \end{pmatrix}$ and by $\widehat{\bA}'= \begin{pmatrix} {\bA}' &  \sigma_\mathsf{min} \I_m \end{pmatrix}$, where $\bA - \bA' = \mathbf{u} \mathbf{v}^{\mathsf T}$. Then $\widehat{\bA}' - \widehat{\bA} = \begin{pmatrix} \mathbf{u} \mathbf{v}^{\mathsf T} & \mathbf{0} \end{pmatrix}$.  Let $\bU_\bC \bSigma_\bC \bV_\bC^{\mathsf{T}}$ be the $\mathsf{SVD}$ of $\bC= \widehat{\bA} \bPhi$ and $\widetilde{\bU}_\bC \widetilde{\bSigma}_\bC \widetilde{\bV}_\bC^{\mathsf{T}}$ be the $\mathsf{SVD}$ of $\widetilde{\bC} =  \widehat{\bA}' \bPhi$. From above discussion, we know that if $\bA - \bA' = \mathbf{u} \mathbf{v}^{\mathsf T}$, then $\bC - \widetilde{\bC} = (1+\alpha)^{1/2}\widetilde{\mathbf{u}} \widetilde{\mathbf{v}}^{\mathsf T}$ for some unit vectors $ \widetilde{\mathbf{u}} $ and $ \widetilde{\mathbf{v}}.$ 
 For notational brevity, in what follows we write $\mathbf{u}$ for $ \widetilde{\mathbf{u}} $ and $\mathbf{v}$ for $ \widetilde{\mathbf{v}}.$


Note that both $\bC$ and $\widetilde{\bC}$ is a full rank matrix because of the construction; therefore $\bC \bC^{\mathsf T}$ is a full dimensional $m \times m$ matrix. This implies that the affine transformation of the multi-variate Gaussian is well-defined (both the covariance $(\bC \bC^{\mathsf T})^{-1}$ has full rank and $\det(\bC \bC^{\mathsf T})$ is non-zero). That is, the {\scshape PDF} of the  distributions of the  rows, corresponding to $\bC$ and $\widetilde{\bC}$, is just a linear transformation of $\cN(0,\I_{m \times m})$. Let $\by \sim \cN(0,1)^t$.  
\begin{align*}
	\PDF_{\bC Y} (\bx) & = \frac{1}{\sqrt{(2\pi)^d \det(\bC  \bC^{\mathsf T})}} e^{(- \frac{1}{2} \bx (\bC  \bC^{\mathsf T})^{-1} \bx^{\mathsf T})} \\
	\PDF_{\widetilde{\bC} Y} (\bx) & = \frac{1}{\sqrt{(2\pi)^d \det(\widetilde{\bC}  \widetilde{\bC}^{\mathsf T})}} e^{(- \frac{1}{2} \bx (\widetilde{\bC}  \widetilde{\bC}^{\mathsf T})^{-1} \bx^{\mathsf T}) }
\end{align*}

 Let  $\varepsilon_0 = \frac{\varepsilon}{\sqrt{4 t \ln (1/\delta)} \log(1/\delta)}$ and  $\delta_0 = {\delta}/{2t},$ 
We prove that every row of the published matrix is $(\varepsilon_0,\delta_0)$ differentially private; the theorem follows from Theorem~\ref{thm:DRV10}.
Let $\bx$ be sampled either from $\cN(0,\bC  \bC^{\mathsf T})$ or  $\cN(0,\widetilde{\bC} \widetilde{\bC}^{\mathsf T})$.
It is straightforward to see that the combination of Claim~\ref{claim1} and Claim~\ref{claim2} proves differential privacy for a row of published matrix. The lemma then follows by an application of~\thmref{DRV10} and our choice of $\varepsilon_0$ and $\delta_0$.

\begin{claim} \label{claim1}
 Let $\bC$ and $\varepsilon_0$ be as defined above. Then 
\begin{align*} e^{- \varepsilon_0}  \leq  \sqrt{\frac{\det(\bC  \bC^{\mathsf T})}{\det(\widetilde{\bC}  \widetilde{\bC}^{\mathsf T})}} \leq  e^{\varepsilon_0}. 
\end{align*}
\end{claim}
\begin{proof}
The claim follows simply as in~\cite{BBDS12} after a slight modification. More concretely, we have $\det(\bC  \bC^{\mathsf T}) = \prod_i \sigma_i^2$, where $\sigma_1 \geq \cdots \geq \sigma_m \geq  \sigma_{\mathsf{min}}(\bC)$ are the singular values of $\bC$. 
 Let $\widetilde{\sigma}_1 \geq \cdots \geq \widetilde{\sigma}_m \geq  \sigma_{\mathsf{min}}(\widetilde{\bC})$ be its singular value  for $\widetilde{\bC}$. The matrix $\mathbf{E}$ has only one singular value $\sqrt{1+\alpha}$. This is because $ \mathbf{E}\mathbf{E}^{\mathsf T} = (1+\alpha) \mathbf{v} \mathbf{v}^{\mathsf T}$. To finish the proof of this claim, we use~\thmref{lidskii}. 

Since the singular values of $\bC - \widetilde{\bC}$ and $\widetilde{\bC} -\bC$ are the same,  Lidskii's theorem (\thmref{lidskii}) gives $\sum_i(\sigma_i - \widetilde{\sigma}_i) \leq \sqrt{1+\alpha}$. 
 Therefore, with probability $1-\delta$,
\begin{align*} 
\sqrt{\prod_{i: \widetilde{\sigma}_i \geq \sigma_i} \frac{\widetilde{\sigma}_i^2}{\sigma_i^2}} 
&= \prod_{i: \widetilde{\sigma}_i \geq \sigma_i} \paren{1 + \frac{\widetilde{\sigma}_i-\sigma_i}{\sigma_i} } 
\leq \exp \paren{\frac{\varepsilon \sum_i (\widetilde{\sigma}_i - \sigma_i)}{32\sqrt{(1+\alpha) t \log (2/\delta)} \log (t/\delta)} } 
\leq e^{\varepsilon_0/2}. \end{align*}

The first inequality holds because $\bPhi \sim \cD_R$ satisfies $(\alpha,\delta)$-{\scshape JLP}  due to the choice of $t$. Since $\bC$ and $\bA$ have same rank, this implies that all the singular values of $\bC$ are  within a $(1\pm \alpha)^{1/2}$ multiplicative factor of $\widehat{\bA}$ due to a result by Sarlos~\cite{Sarlos06}.  In other words, $\sigma_i \geq \sigma_{\mathsf{min}}(\bC)  \geq (1- \alpha)^{1/2} \sigma_\mathsf{min}$.
The case for all $i \in [m]$ when ${\widetilde{\sigma}_i \leq \sigma_i}$ follows similarly as the singular values of $\mathbf{E}$ and $-\mathbf{E}$ are the same. This completes the proof of Claim~\ref{claim1}.
\end{proof}

\begin{claim} \label{claim2}
Let $\bC, \varepsilon_0$, and $\delta_0$ be as defined earlier. Let $\by \sim \cN(0,1)^m$. If $\bx$ is sampled either from $\bC \by$ or $\widetilde{\bC} \by$, then we have
\begin{align*}
\p \sparen{ \left| \bx^{\mathsf T}  (\bC  \bC^{\mathsf T})^{-1} \bx - \bx^{\mathsf T}  (\widetilde{\bC}  \widetilde{\bC}^{\mathsf T})^{-1} \bx\right|   \leq \varepsilon_0} \geq 1 -\delta_0.\end{align*} 
\end{claim}
\begin{proof}
Without any loss of generality, we can assume $\bx = \bC \by$. The case for $\bx = \widetilde{\bC} \by$ is analogous. Let  $\bC - \widetilde{\bC}= \mathbf{v} \mathbf{u}^{\mathsf T}$. Note that $\E[(\bOmega)_{i,j}]=\mathbf{0}^n$ for all $1\leq i,j \leq m$ and $\cov((\bOmega)_{i,j})=1$ if and only if $i=j$. Then
\begin{align*}	
 \bx^{\mathsf T}  (\bC  \bC^{\mathsf T})^{-1} \bx - \bx^{\mathsf T}  (\widetilde{\bC}  \widetilde{\bC}^{\mathsf T})^{-1} \bx 
 &= \bx^{\mathsf T}  (\bC  \bC^{\mathsf T})^{-1} (\widetilde{\bC}  \widetilde{\bC}^{\mathsf T}) (\widetilde{\bC}  \widetilde{\bC}^{\mathsf T})^{-1}  \bx - \bx^{\mathsf T}  (\widetilde{\bC}  \widetilde{\bC}^{\mathsf T})^{-1} \bx \\
	&={\bx^{\mathsf T}   \sparen{(\bC  \bC^{\mathsf T})^{-1} (\bC   \mathbf{u} \mathbf{v}^{\mathsf T} + \mathbf{v}  \mathbf{u}^{\mathsf T}  \widetilde{\bC}^{\mathsf T} ) (\widetilde{\bC}  \widetilde{\bC}^{\mathsf T})^{-1} }   \bx}.  \end{align*}

Using the singular value decomposition of $\bC=\bU_\bC \bSigma_\bC \bV_\bC^{\mathsf T}  $ and $\widetilde{\bC} = \widetilde{\bU}_\bC \widetilde{\bSigma}_\bC \widetilde{\bV}_\bC^{\mathsf T}  $, we have 
\begin{align*} 
&\paren{ \bx ^{\mathsf T}  (\bU_\bC \bSigma_\bC^{-1}\bV_\bC^{\mathsf T}  ) \mathbf{u}} \paren{\mathbf{v}^{\mathsf T}   (\widetilde{\bU}_\bC \widetilde{\bSigma}_\bC^{-2} \widetilde{\bU}_\bC^{\mathsf T}  )  \bx} \\
&\quad+   \paren{\bx^{\mathsf T} (\bU_\bC \bSigma_\bC^{-2}\bU_\bC^{\mathsf T}  )\mathbf{v}} \paren{\mathbf{u}^{\mathsf T}   (\widetilde{\bV}_\bC \widetilde{\bSigma}_\bC^{-1} \widetilde{\bU}_\bC^{\mathsf T}  )  \bx} 
\end{align*}

Since $\bx \sim \bC   \by$, where $\by \sim \cN(0,1)^m$, we can write the above expression as   $\tau_1\tau_2 + \tau_3\tau_4$, where
	\begin{align*} 
		\tau_1&=  \paren{ \by^{\mathsf T}  \bC^{\mathsf T}  (\bU_\bC \bSigma_\bC^{-1}\bV_\bC^{\mathsf T}  ) \mathbf{u} } 
	        & \tau_2= \paren{\mathbf{v}^{\mathsf T}    (\widetilde{\bU}_\bC \widetilde{\bSigma}_\bC^{-2} \widetilde{\bU}_\bC^{\mathsf T}  )  \bC  \by}  \\
		\tau_3&= {\paren{\by^{\mathsf T}  \bC^{\mathsf T} (\bU_\bC \bSigma_\bC^{-2}\bU_\bC^{\mathsf T}  ) \mathbf{v}}}
		& \tau_4={\paren{\mathbf{u}^{\mathsf T}    (\widetilde{\bV}_\bC \widetilde{\bSigma}_\bC^{-1} \widetilde{\bU}_\bC^{\mathsf T}  )\bC  \by}}. 
		\end{align*}

Now since $\| \widetilde{\bSigma}_\bC \|_2, \| \bSigma_\bC \|_2 \geq \sigma_{\min}(\bC)$, plugging in the $\mathsf{SVD}$ of $\bC$ and $\bC - \widetilde{\bC} = \mathbf{v} \mathbf{u}^{\mathsf T}$, and that every term $\tau_i$ in the above expression is a linear combination of a Gaussian, i.e., each term is distributed as per $ \cN(0,\|\tau_i\|^2)$, we calculate $\|\tau_i\|$ as below.
\begin{align*}
\| \tau_1\|_2 &= \| (\bV_\bC \bSigma_\bC \bU_\bC^{\mathsf T} )    (\bU_\bC \bSigma_\bC^{-1} \bV_\bC^{\mathsf T}  ) \mathbf{u} \|_2  \leq \| \mathbf{u}   \|_2 \leq \sqrt{1+\alpha}, \\
\| \tau_2\|_2 &= \| \mathbf{v}^{\mathsf T}   (\widetilde{\bU}_\bC \widetilde{\bSigma}_\bC^{-2} \widetilde{\bU}_\bC^{\mathsf T}  )  (\widetilde{\bU}_\bC \widetilde{\bSigma}_\bC \widetilde{\bV}_\bC^{\mathsf T}   - \mathbf{v}\mathbf{u}^{\mathsf T}  )\|_2  \\ 
		& \quad \leq \| \mathbf{v}^{\mathsf T}   (\widetilde{\bU}_\bC \widetilde{\bSigma}_\bC^{-2} \widetilde{\bU}_\bC^{\mathsf T}  )    \widetilde{\bU}_\bC \widetilde{\bSigma}_\bC  \widetilde{\bU}_\bC^{\mathsf T}  \|_2 
		+ \| \mathbf{v}^{\mathsf T}   (\widetilde{\bU}_\bC \widetilde{\bSigma}_\bC^{-2} \widetilde{\bU}_\bC^{\mathsf T}  )    \mathbf{v}\mathbf{u}^{\mathsf T}   \|_2   \leq \frac{1}{ \sigma_{\mathsf{min}}(\bC)} + \frac{\sqrt{1+\alpha} }{ \sigma_{\mathsf{min}}^2(\bC)},  \\
\| \tau_3\|_2 &= \| (\bV_\bC \bSigma_\bC \bU_\bC^{\mathsf T} )(\bU_\bC \bSigma_\bC^{-2}\bU_\bC^{\mathsf T}  )\mathbf{v} \|_2 
	\leq  \| \bSigma_\bC^{-1} \|_2 
\leq \frac{ 1}{ \sigma_{\mathsf{min}}(\bC)},  \\
\| \tau_4\|_2 
	&= \| \mathbf{u}^{\mathsf T}   (\widetilde{\bV}_\bC \widetilde{\bSigma}_\bC^{-1} \widetilde{\bU}_\bC^{\mathsf T}  )   (\widetilde{\bU}_\bC \widetilde{\bSigma}_\bC \widetilde{\bV}_\bC^{\mathsf T}   - \mathbf{v}\mathbf{u}^{\mathsf T}  ) \|_2  \\ 
		& \quad \leq \|  \mathbf{u}^{\mathsf T}   (\widetilde{\bV}_\bC \widetilde{\bSigma}_\bC^{-1} \widetilde{\bU}_\bC^{\mathsf T}  )  (\widetilde{\bU}_\bC \widetilde{\bSigma}_\bC 
\widetilde{\bV}_\bC^{\mathsf T}   \|_2 
		+ \|\mathbf{u}^{\mathsf T}   (\widetilde{\bV}_\bC \widetilde{\bSigma}_\bC^{-1} \widetilde{\bU}_\bC^{\mathsf T}  )  \mathbf{v}  \|_2   \leq \sqrt{1+\alpha}  + \frac{\sqrt{1+\alpha}}{ \sigma_{\mathsf{min}}(\bC)}.
\end{align*}
 
 Using the concentration bound on the Gaussian distribution, each term, $\tau_1,\tau_2,\tau_3$, and $\tau_4$, is less than $\|\tau_i\| \ln (4/\delta_0)$ with probability $1 - \delta_0/2$.  The second claim follows because with probability $1-\delta_0$,
 \begin{align*}
 \left| \bx^{\mathsf T}  (\bC  \bC^{\mathsf T})^{-1} \bx - \bx^{\mathsf T}  (\widetilde{\bC}^{\mathsf T}   \widetilde{\bC})^{-1} \bx\right| 
  \leq  2\paren{\frac{\sqrt{1+\alpha}}{ \sigma_{\mathsf{min}}(\bC)} + \frac{{1+\alpha}}{ \sigma_{\mathsf{min}}^2(\bC)}} \ln (4/\delta_0)  \leq \varepsilon_0 ,
\end{align*} 
where the second inequality follows from the choice of $\sigma_\mathsf{min}$ and the fact that $\sigma_\mathsf{min}(\bC) \geq (1-\alpha)^{1/2} \sigma_\mathsf{min}$.
\end{proof}
\lemref{low_space_private} follows by combining Claims~\ref{claim1} and~\ref{claim2}.




\section{Related Problems Studied in Previous Works} \label{app:related}
Differential privacy was  introduced by Dwork {\it et al.}~\cite{DMNS06}. 
Since then, many algorithms for preserving differential privacy have been proposed in the literature (see, Dwork and Roth~\cite{DR14}). 
Dwork  {\it et al.}~\cite{DNPR10} first considered streaming algorithms with privacy under the model of {\em pan-privacy}, where the internal state is known to the adversary. 
Subsequently, there have been some works on online private learning~\cite{DTTZ14,JKT12,TS13} for various tasks.  
There are some recent works on differentially private $\lra$ as well. Blum  {\it et al.}~\cite{BDMN05} first studied this problem in the Frobenius norm. This was improved by Hardt and Roth~\cite{HR12}  under the low coherence assumption. Upadhyay~\cite{Upadhyay14} later showed that one can make the two-pass algorithm of Hardt and Roth~\cite{HR12} single-pass. differentially private $\lra$ has been studied in the spectral norm as well by many works~\cite{CSS12,KT13,HR13,HP14}.
Recently, Dwork  {\it et al.}~\cite{DTTZ14} gave a tighter analysis of the algorithm of Blum  {\it et al.}~\cite{BDMN05}  and used it to give a private online algorithm for covariance matrices. We now define the problem studied in each of the work listed above and draw the contrast with Problem~\ref{prob:private_spectral}.

\paragraph {\bf Spectral low-rank approximation} 
Kapralov and Talwar~\cite{KT13}, Hardt and Roth~\cite{HR13},  Hardt and Price~\cite{HP14}, and Upadhyay~\cite{Upadhyay14} studied low-rank approximation. 
\begin{prob} 
\label{prob:private_spectral} (Approximation with respect to the spectral norm).
Given parameters $\alpha,\beta, \tau$, a private $m \times n$ matrix $\bA$  (where $m \ll n$) and the target rank $k$, find a rank-$k$ matrix ${\bB}$ such that 
\begin{align*}  \p \sparen{ \| \bA -{\bB} \|_2 \leq \gamma  \Delta_k(\bA) +\zeta } \geq 1-\beta. \end{align*} 
Hardt and Price~\cite{HP14} and Jiang {\it et al.}~\cite{JXZ15} consider two matrices as neighboring if they differ in exactly one entry by at most $1$. Upadhyay~\cite{Upadhyay14} considered two matrices are neighboring if their differ in at most one row or column of norm $1$.  Kapralov and Talwar~\cite{KT13} considered two matrices as neighboring if the difference of their  spectral norm is at most $1$. The result of Kapralov and Talwar~\cite{KT13} and Jiang {\it et al.}~\cite{JXZ15} is for $\varepsilon$-differential privacy (i.e., $\delta=0$).
\end{prob}
\medskip \noindent \textbf{Difference from this paper.} We consider low-rank factorization with respect to the spectral norm while Problem~\ref{prob:private_spectral} studied only low-rank approximation. Moreover, granularity of privacy we consider is more general than theirs.

\paragraph {\bf Approximating the right singular vectors} 
Dwork {\it et al.}~\cite{DTTZ14} studied the following problem.
\begin{prob} \label{prob:private_DTTZ}
Given parameters $\alpha,\beta, \tau$ and an $m \times n$  private matrix $\bA$ (where $m \gg n)$), compute a rank-$k$ matrix ${\bB}$ such that with probability $1-\beta$
\begin{align*}   { {\| \bA^{\mathsf T} \bA -{\bB}\|_2 } \leq  \min_{{\sf rank}(\bB_k) \leq k} {\|  \bA^{\mathsf T} \bA - \bB_k\| } +\zeta }.\end{align*}  
\end{prob}
Dwork {\it et al.}~\cite{DTTZ14} consider two matrices neighbouring if they differ by at most one row. They further assume that the rows are normalized. 

\medskip \noindent \textbf{Difference from this paper.} We consider low-rank factorization of both the right and the left singular vectors while Problem~\ref{prob:private_DTTZ} studied  low-rank ``approximation" of the right singular vectors. Moreover, granularity of privacy we consider is more general than theirs.

\section{Useful Facts}

We use the notation $\rad(p)$ to denote a distribution with support $\pm 1$ such that $+1$ is sampled with probability  $p$ and $-1$ is sampled with probability  $1-p$. 
An $n \times n$ Walsh-Hadamard matrix $\bH_n$ is constructed recursively as follows:
\begin{align*}  \bH_n = \begin{pmatrix}  \bH_{n/2} & \bH_{n/2} \\  \bH_{n/2} & -\bH_{n/2} \end{pmatrix}~\text{and}~\bH_1 :=1.\end{align*}
A randomized Walsh-Hadamard matrix $\bW_n$ is formed by multiplying $\bH_n$ with a diagonal matrix whose diagonal entries are picked i.i.d. from $\rad(1/2).$ We drop the subscript $n$ where it is clear from the context. A $r \times n$ subsampled randomized Hadamard matrix is constructed by multiplying $\bPi_{1..r}$ from the left to a randomized Hadamard matrix, where $\bPi_{1..r}$ is the matrix formed by the first $r$ rows of a random permutation matrix. 

\begin{lemma} \label{lem:SRHTinverse}
Let $\bS$ be a $v \times m$ subsampled randomized Hadamard matrix, where $v\leq m$ and $\bN_2 \in \R^{v \times n}$. Then we have, 
\begin{align*} \|  \bS^\dagger \bN_2 \|_2 = \| \bN_2 \|_2.\end{align*}  
\end{lemma} 
 \begin{proof}
One way to look at the action of $\bS$ when it is a subsampled Hadamard transform is that it is a product of matrices $\bW$ and $\bPi_{1..v}$, where $\bPi_{1..v}$ is the matrix formed by the first $r$ rows of a random permutation matrix and $\bW$ is a randomized Walsh-Hadamard matrix formed by multiplying a Walsh-Hadamard matrix with a diagonal matrix whose non-zero entries are picked i.i.d. from $\rad(1/2)$. 

 Since $\bW\bD$ has orthonormal rows, $\bS^\dagger=(\bPi_{1..v} \bW \bD)^\dagger= (\bW \bD)^{\mathsf T} (\bPi_{1..v})^{\dagger}$.  
\begin{align*} 
\|  \bS^\dagger \bN \|_2 &= \| ({\bPi}_{1..v} \bW \bD)^\dagger \bN \|_2 = \| (\bW \bD)^{\mathsf T} {\bPi}_{1..v}^{\dagger} \bN \|_2 \\ 
	& = \| {\bPi}_{1..v}^{\dagger} \bN \|_2. 
\end{align*}

Using the fact that $\bPi_{1..v}$ is a full row rank matrix and ${\bPi}_{1..v}  {\bPi}_{1..v}^{\mathsf T}$ is an identity matrix, we have 
${\bPi}_{1..v}^{\dagger} = {\bPi}_{1..v}^{\mathsf T} ( \widehat{\bPi}_{1..v}  \bPi_{1..v}^{\mathsf T})^{-1} =  {\bPi}_{1..v}^{\mathsf T}. $ The result follows. 
 \end{proof}

\medskip \noindent \textbf{Differential privacy.}
We use the following results about differential privacy in this paper.

\begin{lemma} \label{lem:post}  (Post-processing~\cite{DKMMN06}).
Let $\mathfrak{M}(\mathbf{D})$ be an $(\varepsilon, \delta)$-differential private mechanism for a database $\mathbf{D}$, and let $h$ be any function, then any mechanism $\mathfrak{M}':=h(\mathfrak{M}(\mathbf{D}))$ is also $(\varepsilon, \delta)$-differentially private. 
\end{lemma}

\begin{theorem} \label{thm:DRV10}   (Composition~\cite{DRV10}). 
	Let $\varepsilon_0, \delta_0 \in (0,1)$, and $\delta'>0$. If $\mathfrak{M}_1, \cdots , \mathfrak{M}_\ell$ are each $(\varepsilon, \delta)$-differential private mechanism, then the mechanism $$\mathfrak{M}(\mathbf{D}):= (\mathfrak{M}_1(\mathbf{D}), \cdots , \mathfrak{M}_\ell(\mathbf{D}))$$ releasing the concatenation of each algorithm is $(\varepsilon', \ell \delta+0+\delta')$-differentially private for $\varepsilon' < \sqrt{2\ell \ln (1/\delta')}\varepsilon_0 + 2\ell \varepsilon_0^2$.
\end{theorem}

\begin{theorem} (Gaussian mechanism~\cite{DKMMN06}.) \label{thm:gaussian}
Let $\bx, \by \in \R^n$ be any two vectors such that $\| \bx - \by \|_2 \leq c$. Let $\rho = c\varepsilon^{-1} \sqrt{\log(1/\delta)}$ and $\mathbf{g} \sim \cN(0,\rho^2)^n$ be a vector with each entries sampled i.i.d. Then for any $\mathbf{s} \subset \R^n$, $  \p [\bx + \mathbf{g} \in \mathbf{s}] \leq e^{\varepsilon} \p [\by + \mathbf{g} \in \mathbf{s}] + \delta. $
\end{theorem}

\medskip \noindent \textbf{Properties of Gaussian distribution.} We need the following property of a random Gaussian matrices. 

\begin{fact} \label{fact:gaussian} (\cite{JL84,Sarlos06})
Let $\mathbf{P} \in \R^{m \times n}$ be a matrix of rank $r$ and $\mathbf{Q} \in \R^{m \times n'}$ be an $m \times n'$ matrix. Let $\cD$ be a distribution of matrices over $\R^{t \times n}$ with entries  sampled i.i.d. from $\cN(0,1/t)$. Then there exists a $t = O(r/\alpha \log(r/\beta) )$ such that $\cD$ is an $(\alpha,\beta)$-subspace embedding for generalized regression.
\end{fact}

\medskip \noindent \textbf{Matrix Theory.} We use some basic results from matrix theory. 
\begin{theorem} \label{thm:Weyl} (Weyl's perturbation theorem)
For any $m \times n$ matrices $\mathbf{P}, \mathbf{Q}$, we have $| \sigma_i(\mathbf{P} + \mathbf{Q}) - \sigma_i(\mathbf{P})| \leq \| \mathbf{Q} \|_2$, where $\sigma_i(\cdot)$ denotes the $i$-th singular value and $\| \mathbf{Q} \|_2$ is the spectral norm of the matrix $\mathbf{Q}$.
\end{theorem}


\begin{fact}  \label{fact:normequivalence}
Let $\bA$ be a rank-$r$  matrix, then $\| \bA  \|_2 \leq \| \bA\|_F \leq \sqrt{r} \| \bA \|_2$.
\end{fact}

We also need the following results for the privacy proof.
\begin{theorem} \label{thm:lidskii} (Lidskii Theorem). 
Let $\bA, \mathbf{B}$ be $n \times n$ Hermittian matrices. Then for any choice of indices $1 \leq i_1 \leq \cdots \leq i_k \leq n$,  
\begin{align*}  \sum_{j=1}^k  \lambda_{i_j}(\bA + \mathbf{B}) \leq \sum_{j=1}^k  \lambda_{i_j}(\bA) + \sum_{j=1}^k  \lambda_{i_j}( \mathbf{B}), \end{align*}
where $\set{\lambda_i(\bA)}_{i=1}^n$ are the eigen-values of $\bA$ in decreasing order.
\end{theorem}

We need the following theorem in our proofs.

\begin{lemma} (Sarlos~\cite{Sarlos06}) \label{lem:sarlos}
Let $\bPhi$ be a matrix that satisfies $(\alpha,\beta)$-subspace embedding for $\I$, then $(1- \alpha) \leq \| \bPhi \|_2 \leq (1+\alpha)$ with probability at least $1-\beta$.
\end{lemma}

\begin{lemma} \label{lem:CW09}
Given integer $k$ and $\alpha,\delta >0$, there is $p = O(k \log(1/\delta)/\alpha)$ such that if S is an $p \times n$ $(\alpha,\beta)$-subspace embedding matrix, then for $n \times k$ matrix $\bU$ with orthonormal columns, with probability at least $1 - \delta$, the spectral norm $\| \bU^\sT \bPhi \bPhi^\sT \bU - \I \|_2 \leq \alpha$.
\end{lemma}

\begin{theorem} \label{thm:FT07} {(Sou and Rantzer~\cite{SR12}).}
Let matrices $\mathbf{O} \in \R^{m \times n}, \mathbf{L} \in \R^{m \times p}$ and $\mathbf{R} \in \R^{q \times n}$ be given.  Then 
 $$\mathbf{L}^\dagger [ \bU_\mathbf{L} \bU_\mathbf{L}^{\mathsf T} \mathbf{O} \bV_\mathbf{R} \bV_\mathbf{R}^{\mathsf T}]_k \mathbf{R}^\dagger =  \argmin_{\bX, r(\bX) \leq k} \| \mathbf{O} - \mathbf{L} \bX \mathbf{R} \|_2, $$ where $\mathbf{L} := \bU_\mathbf{L} \bSigma_\mathbf{L} \bV_\mathbf{L}^{\mathsf T}$ and $\mathbf{R} := \bU_\mathbf{R} \bSigma_\mathbf{R} \bV_\mathbf{R}^{\mathsf T}$.
\end{theorem}

\end{appendix}

\end{document}